\documentclass[runningheads]{llncs}
\usepackage[utf8]{inputenc}
\usepackage[T1]{fontenc}

\usepackage{graphicx}
\usepackage{hyperref}

\title{Derivation of Heard-Of Predicates From Elementary Behavioral Patterns}

\author{Adam Shimi \and Aurélie Hurault \and Philippe Queinnec}

\authorrunning{A. Shimi, A. Hurault, P. Queinnec} 

\institute{IRIT -- Université de Toulouse,
                    2 rue Camichel,
                    F-31000 Toulouse, France\\
           \email{\{firstname\}.\{lastname\}@irit.fr}}

\usepackage{amsmath}
\usepackage{amssymb}
\usepackage{enumitem}
\usepackage{tabularx}

\spnewtheorem*{theorem*}{Theorem}{\bfseries}{\itshape}
\spnewtheorem*{lemma*}{Lemma}{\bfseries}{\itshape}
\spnewtheorem*{corollary*}{Corollary}{\bfseries}{\itshape}

\newcommand{\combi}{\bigotimes}

\usepackage{todonotes}
\usepackage{color}
\usepackage{xspace}
\usepackage{pifont}

\begin{document}

\maketitle

\begin{abstract}
    There are many models of distributed computing, and no unifying
    mathematical framework for considering them all.
    One way to sidestep this issue is to start
    with simple communication and fault models, and use them
    as building blocks to derive the complex models studied
    in the field. We thus define operations like union, succession
    or repetition, which makes it easier to build complex models from simple
    ones while retaining expressivity.

    To formalize this approach, we abstract away the complex models and
    operations in the Heard-Of model. This model relies on
    (possibly asynchronous) rounds; sequence of digraphs, one for each round,
    capture which messages sent at a given round are received before
    the receiver goes to the next round.
    A set of sequences, called a heard-of predicate,
    defines the legal communication behaviors --
    that is to say, a model of communication.
    Because the proposed operations behave well with this
    transformation of operational models into heard-of predicates,
    we can derive bounds, characterizations, and implementations
    of the heard-of predicates for the constructions.

    \keywords{Message-passing \and Asynchronous Rounds \and Failures \and Heard-Of Model}
\end{abstract}

\section{Introduction}
\label{sec:intro}

    \subsection{Motivation}
    \label{subsec:motiv}

      Let us start with a round-based distributed algorithm; such an
      algorithm is quite common in the literature, especially
      in fault-tolerant settings. We want to formally verify this algorithm
      using the methods of our choice: proof-assistant, model-checking,
      inductive invariants, abstract interpretation\ldots{}
      But how are we supposed to model the context in which the algorithm will
      run? Even a passing glance at the distributed computing literature
      shows a plethora of models defined in the mixture of english and
      mathematics.

      Thankfully, there are formalisms for abstracting round-based models
      of distributed computing. One of these is the Heard-Of model of Charron-Bost
      and Schiper~\cite{CharronBostHO}; it boils down the communication model
      to a description of all accepted combinations of received messages.
      Formally, this is done by considering communications graphs, one for each round, and
      taking the sets of infinite sequences of graphs that are allowed by
      the model. Such a set is called a heard-of predicate, and captures
      a communication model.

      An angle of attack for verification is therefore to find the heard-of predicate
      corresponding to a real-world environment, and use the techniques from
      the literature to verify an algorithm for this heard-of predicate.
      But which heard-of predicate should be used? What is the "right" predicate
      for a given environment? For some cases, the predicates are given
      in Charron-Bost and Schiper~\cite{CharronBostHO}; but this does not solve
      the general case.

      Actually, the answer is quite subtle. This follows from a fundamental part of the
      Heard-Of model: communication-closedness~\cite{ElradDecomp}.
      This means that for $p$ to use a message from $q$ at round $r$, $p$ must receive it
      before or during its own round $r$. And thus, knowing whether
      $p$ receives the message from $q$ at the right round or not
      depends on how $p$ waits for messages. That is,
      it depends on the specifics of how rounds are implemented on top of it.

      Once again, the literature offers a solution:
      Shimi et al.~\cite{ShimiOPODIS18} propose to first find
      a delivered predicate -- a description of which messages will eventually
      be delivered, without caring about rounds --,
      and then to derive the heard-of predicate from it. This derivation
      explicitly studies strategies, the aforementioned rules for how processes
      waits for messages before changing round.

      But this brings us back to square one: now we are looking
      for the delivered predicate corresponding to a real-world model, instead
      of the heard-of predicate.
      Basic delivered predicates for elementary failures are easy to find,
      but delivered predicates corresponding to combinations of
      failures are often not intuitive.

      In this paper, we propose a solution to this problem: building a
      complex delivered predicate from simpler ones we already know.
      For example, consider a system where one process can crash and may
      recover later, and another process can definitively crash.
      The delivered predicate for at most one crash is $PDel^{crash}_1$, and
      the predicate where all the messages are delivered is $PDel^{total}$.
      Intuitively, a process that can crash and necessarily recover is
      described by the behavior of $PDel^{crash}_1$ followed by
      the behavior of $PDel^{total}$. We call this the succession of these predicates,
      and write it $PDel^{recover}_1 \triangleq PDel^{crash}_1 \leadsto PDel^{total}$.
      In our system, the crashed process may never recover:
      hence we have either the behavior of $PDel^{recover}_1$
      or the behavior of $PDel^{crash}_1$. This amounts to a union (or a disjunction);
      we write it $PDel^{canrecover}_1 \triangleq
      PDel^{recover}_1 \cup PDel^{crash}_1$.
      Finally, we consider a potential irremediable crash, additionally to
      the previous predicate. Thus we want the behavior of $PDel^{crash}_1$ and
      the behavior of $PDel^{canrecover}_1$.
      We call it the combination (or conjunction) of these predicates, and
      write it $PDel^{crash}_1 \combi PDel^{canrecover}_1$
      The complete system is thus described by $PDel^{crash}_1 \combi ((PDel^{crash}_1
      \leadsto PDel^{total}) \cup PDel^{crash}_1)$.
      In the following, we will also introduce an operator $\omega$ to express
      repetition. For example, a system where, repeatedly, a process can crash and recover is
      $(PDel^{crash}_1 \leadsto PDel^{total})^\omega$.

      Lastly, the analysis of the resulting delivered predicate can be bypassed:
      its heard-of predicate arises from our operations applied
      to the heard-of predicates of the elementary building blocks.

    \subsection{Related Work}
    \label{subsec:related}

        The heard-of model was proposed by Charron-Bost and
        Schiper~\cite{CharronBostHO} as a combination of the
        ideas of two previous work. First, the concept of a fault
        model where the only information is which
        message arrives, from Santoro and Widmayer~\cite{SantoroLoss};
        and second, the idea of abstracting failures in a round
        per round fashion, from Gafni~\cite{GafniRRFD}. Replacing
        the operational fault detectors of Gafni with the fault model
        of Santoro and Widmayer gives the heard-of model.

        This model was put to use in many ways. Obviously computability
        and complexity results were proven: new algorithms for consensus
        in the original paper by Charron-Bost and Schiper~\cite{CharronBostHO};
        characterizations for consensus solvability by
        Coulouma et al.~\cite{CouloumaConsensus} and Nowak et
        al.~\cite{NowakTopoConsensus}; a characterization
        for approximate consensus solvability by Charron-Bost et
        al.~\cite{CharronBostApprox}; a study of $k$ set-agreement
        by Biely et al.~\cite{BielyKSet}; and more.

        The clean mathematical abstraction of the heard-of model also
        works well with formal verification. The rounds provide
        structure, and the reasoning can be less operational than
        in many distributed computing abstractions. We thus have
        a proof assistant verification of consensus algorithms
        in Charron-Bost et al.~\cite{CharronBostHOL}; cutoff
        bounds for the model checking of consensus algorithms
        by Mari{\'{c}} et al.~\cite{MaricCutoff};
        a DSL to write code following the structure of the heard-of model
        and verify it with inductive invariants by Dr\u{a}goi et
        al.~\cite{DragoiPsync}; and more.

    \subsection{Contributions}
    \label{subsec:results}

      The contributions of the paper are:

      \begin{itemize}
        \item A definition of operations on delivered predicates
          and strategies, as well as examples using them
          in Section~\ref{sec:ops}.
        \item The study of oblivious strategies, the strategies only looking
          at messages for the current round, in Section~\ref{sec:carefr}.
          We provide a technique to extract a strategy dominating the oblivious
          strategies of the built predicate
          from the strategies of the initial predicates;
          exact computations of the generated heard-of predicates;
          and a sufficient condition on the building blocks for the
          result of operations to be dominated by an oblivious strategy.
        \item The study of conservative strategies, the strategies
          looking at everything but messages from future rounds,
          in Section~\ref{sec:cons}. We provide
          a technique to extract a strategy dominating the conservative
          strategies of the build predicate
          from the strategies of the initial predicates;
          upper bounds on the generated heard-of predicates;
          and a sufficient condition on the building blocks for the
          result of operations to be dominated by a conservative strategy.
      \end{itemize}

      Due to size constraints, many of the complete proofs are not
      in the paper itself, and can be found in the appendix.

\section{Operations and Examples}
\label{sec:ops}

  \subsection{Basic concepts}

    We start by providing basic definitions and intuitions. The model we consider
    proceed by rounds, where processes send messages tagged with a round number,
    wait for some messages with this round number, and then compute the next
    state and increment the round number.
    $\mathbb{N}^*$ denotes the non-zero naturals.

    \begin{definition}[Collections and Predicates]
        Let $\Pi$ a set of processes.
        An element of $(\mathbb{N}^* \times \Pi) \mapsto \mathcal{P}(\Pi)$ is
        either a \textbf{Delivered collection} $c$ or
        a \textbf{Heard-Of collection} $h$ for $\Pi$, depending on the context.
        $c_{tot}$ is the total collection such that $\forall r > 0,
        \forall p \in \Pi: c_{tot}(r,p) = \Pi$.

        An element of
        $\mathcal{P}((\mathbb{N}^* \times \Pi) \mapsto \mathcal{P}(\Pi))$
        is either a \textbf{Delivered predicate} $PDel$
        or a \textbf{Heard-Of predicate} $PHO$
        for $\Pi$. $\mathcal{P}_{tot} = \{c_{tot}\}$
        is the total delivered predicate.
    \end{definition}

    For a heard-of collection $h$, $h(r,p)$ are the senders of messages
    for round $r$ that $p$ has received at or before its round $r$,
    and thus has known while at round $r$.
    For a delivered collection $c$, $c(r,p)$ are the senders of messages
    for round $r$ that $p$ has received, at any point in time.
    Some of these messages may have arrived early, before $p$ was at $r$,
    or too late, after $p$ has left round $r$.
    $c$ gives an operational point of view (which messages arrive),
    and $h$ gives a logical point of view (which messages are used).

    \begin{remark}
        We also regularly use the "graph-sequence" notation
        for a collection~$c$. Let $Graphs_{\Pi}$ be the set of graphs whose nodes are
        the elements of $\Pi$. A collection $gr$ is an element of
        $(Graphs_{\Pi})^{\omega}$.
        We say that $c$ and $gr$ represent the same collection
        when $\forall r > 0, \forall p \in \Pi: c(r,p)
        = In_{gr[r]}(p)$, where $In(p)$ is the incoming vertices of $p$.
        We will usually not define two collections but use one
        collection as both kind of objects; the actual type being used
        in a particular expression can be deduced from the operations
        on the collection. For example $c[r]$ makes sense for a sequence
        of graphs, while $c(r,p)$ makes sense for a function.
    \end{remark}

    In an execution, the local state of a process is the pair of its current
    round and all the received messages up to this point. We disregard
    any local variable, since our focus is on which messages to wait for.
    A message is represented
    by a pair $\langle round, sender\rangle$. For a state $q$, and a round $r > 0$,
    $q(r)$ is the set of peers from which the process has received a message
    for round $r$.

    \begin{definition}[Local State]
        Let $Q = \mathbb{N}^* \times \mathcal{P}(\mathbb{N}^* \times \Pi)$.
        Then $q \in Q$ is a \textbf{local state}.

        For $q = \langle r, mes \rangle$, we write $q.round$ for $r$,
        $q.mes$ for $mes$ and $\forall i > 0: q(i) \triangleq \{k \in \Pi \mid \langle i, k \rangle \in q.mes\}$.
    \end{definition}

    We then define strategies, which constrain the behavior of processes.
    A strategy is a set of states from which a process
    is allowed to change round.
    It captures rules like "wait for at least $F$ messages from the
    current round", or "wait for these specific messages".
    Strategies give a mean to constrain executions. 

    \begin{definition}[Strategy]
        $f \in \mathcal{P}(Q)$ is a \textbf{strategy}.
    \end{definition}

    \subsection{Definition of Operations}
    \label{subsec:defOps}

        We can now define operations on predicates and their corresponding strategies.
        The intuition behind these operations is the following:
        \begin{itemize}
          \item The union of two delivered predicates is equivalent
            to an OR on the two communication behaviors. For example,
            the union of the delivered predicate for one crash at round $r$
            and of the one for one crash at round $r+1$ gives
            a predicate where there is either a crash at round $r$
            or a crash at round $r+1$.
          \item The combination of two behaviors takes
            every pair of collections, one from each predicate, and
            computes the intersection of the graphs at each round. Meaning, it
            adds the loss of messages from both, to get both behaviors
            at once. For example, combining $PDel^{crash}_1$ with itself gives
            $PDel^{crash}_2$, the predicate with at most two crashes.
            Although combination intersects graphs round by round
            in a local fashion, it actually combines two collections globally,
            and thus can combine several global predicates like hearing
            from a given number of process during the whole execution.
          \item For succession, the system
            starts with one behavior, then switch to another. The definition is
            such that the first behavior might never happen, but the second one must
            appear.
          \item Repetition is the next logical step after succession: instead of following
            one behavior with another, the same behavior is repeated again and again.
            For example, taking the repetition of at most one crash results
            in a potential infinite number of crash-and-restart,
            with the constraint of having at most one crashed process at any time.
        \end{itemize}

        \begin{definition}[Operations on predicates]
            Let $P_1, P_2$ be two delivered or heard-of predicates.
            \begin{itemize}
              \item The \textbf{union} of $P_1$ and $P_2$ is $P_1 \cup P_2$.
              \item The \textbf{combination} $P_1 \combi P_2 \triangleq
                \{c_1 \combi c_2 \mid c_1 \in P_1, c_2 \in P_2 \}$,
                where for $c_1$ and $c_2$ two collections,
                $\forall r > 0, \forall p \in \Pi: (c_1 \combi c_2)(r,p) =
                c_1(r,p) \cap c_2(r,p)$.
              \item The \textbf{succession} $P_1 \leadsto P_2 \triangleq
                \bigcup\limits_{c_1 \in P_1, c_2 \in P_2} c_1 \leadsto c_2$,\\
                with $c_1 \leadsto c_2 \triangleq \{ c \mid
                \exists r \geq 0 : c = c_1[1,r].c_2\}$.
              \item The \textbf{repetition} of $P_1$, $(P_1)^{\omega} \triangleq
                \{c \mid \exists (c_i)_{i \in \mathbb{N}^*},
                \exists (r_i)_{i \in \mathbb{N}^*}:
                r_1 = 0 \land
                \forall i \in \mathbb{N}^*:
                (c_i \in P_1 \land r_{i} < r_{i+1} \land
                c[r_i+1,r_{i+1}]=c_i[1,r_{i+1} - r_i]) \}$.
            \end{itemize}
        \end{definition}

        For all operations on predicates, we provide an analogous one for
        strategies. We show later that strategies for the delivered predicates,
        when combined by the analogous operation, retain important
        properties on the result of the operation on the predicates.

        \begin{definition}[Operations on strategies]
            Let $f_1, f_2$ be two strategies.
            \begin{itemize}
              \item Their \textbf{union} $f_1 \cup f_2 \triangleq$ the strategy such that
                $\forall q$ a local state: $(f_1 \cup f_2)(q)
                \triangleq f_1(q) \lor f_2(q)$.
              \item Their \textbf{combination} $f_1 \combi f_2 \triangleq
                \{ q_1 \combi q_2 \mid q_1 \in f_1 \land q_2 \in f_2
                \land q_1.round = q_2.round\}$,
                where for $q_1$ and $q_2$ at the same round $r$,
                $q_1 \combi q_2 \triangleq
                \langle r
                    \{ \langle r', k \rangle
                    \mid r' > 0 \land k \in q_1(r') \cap q_2(r')\}
                \rangle$
              \item Their \textbf{succession} $f_1 \leadsto f_2 \triangleq
                f_1 \cup f_2 \cup \{q_1 \leadsto q_2 \mid
                q_1 \in f_1 \land q_2 \in f_2 \}$
                where $q_1 \leadsto q_2 \triangleq
                \left\langle\begin{array}{l}
                q_1.round+q_2.round, \\
                    \left\{ \langle r, k \rangle
                    \mid r > 0 \land
                    \left(
                    \begin{array}{ll}
                        k \in q_1(r) & \text{if } r \leq q_1.round\\
                        k \in q_2(r-q_1.round) & \text{if } r > q_1.round\\
                    \end{array}
                    \right) \right\}
                \end{array}
                \right\rangle$
              \item The \textbf{repetition} of $f_1$, $f_1^{\omega} \triangleq
                \{q_1 \leadsto q_2 \leadsto ... \leadsto q_k \mid
                k \geq 1 \land q_1,q_2,...,q_k \in f_1$\}.
            \end{itemize}
        \end{definition}

      The goal is to derive new strategies
      for the resulting model by applying operations on strategies for the
      starting models. This allows, in some cases, to bypass
      strategies, and deduce the Heard-Of predicate for a given Delivered
      predicate from the Heard-Of predicates of its building blocks.

    \subsection{Executions and Domination}
    \label{subsec:prelim}

      Before manipulating predicates and strategies, we need
      to define what is an execution: a specific ordering
      of events corresponding to a delivered collection.
      An execution is an infinite sequence of either delivery
      of messages ($deliver(r,p,q)$), change to the next
      round ($next_j$), or a deadlock ($stop$).
      Message sending is implicit
      after every change of round.
      An execution must satisfy three rules:
      no message is delivered before it is sent, no message
      is delivered twice, and once there is a $stop$, the rest of
      the sequence can only be $stop$.

      \begin{definition}[Execution]
          Let $\Pi$ be a set of $n$ processes.
          Let the set of transitions $T = \{ \textit{next}_j \mid j \in \Pi \} \cup
          \{ \textit{deliver}(r,k,j) \mid r \in \mathbb{N}^* \land k,j \in \Pi\}
          \cup \{ stop \}$. \textit{next}$_j$ is the transition for $j$ changing round,
          \textit{deliver}$(r,k,j)$ is the transition for the delivery
          to $j$ of the message sent by $k$ in round $r$, $stop$ models a deadlock.
          Then, $t \in T^{\omega}$ is an \textbf{execution}
          $\triangleq$
          \begin{itemize}
              \item \textbf{(Delivery after sending)}\\
                $\forall i \in \mathbb{N}:
                t[i] = deliver(r,k,j) \implies
                \mathbf{card}(\{l \in [0,i[ \mid t[l] = next_k\}) \geq r-1$
              \item \textbf{(Unique delivery)}\\
                  $\forall \langle r, k, j \rangle \in
                  (\mathbb{N}^* \times \Pi \times \Pi):
                  \mathbf{card}(\{i \in \mathbb{N} \mid t[i] = deliver(r,k,j)\}) \leq 1$
              \item \textbf{(Once stopped, forever stopped)}\\
                $\forall i \in \mathbb{N}:
                t[i] = stop \implies \forall j \ge i : t[j] = stop$
          \end{itemize}

          Let \textit{c} be a delivered collection.
          Then, $execs(c)$, the \textbf{executions of}
          \textit{c} $\triangleq$\\
          $\left\{
              t \textit{ an execution} ~\middle|~
              \begin{array}{l}
                  \forall \langle r, k, j \rangle
                  \in \mathbb{N}^* \times \Pi \times \Pi:\\
                  \quad(
                      k \in c(r,j)
                      \land \mathbf{card}(\{ i \in \mathbb{N} \mid t[i] = next_k\}) \geq r-1
                  )\\
                  \quad\iff\\
                  \quad(
                      \exists i \in \mathbb{N}: t[i] = deliver(r,k,j)
                  )
              \end{array}
              \right\}$

          For a delivered predicate \textit{PDel},
          $execs(PDel) \triangleq \{execs(c) \mid c \in PDel\}$.

          Let $t$ be an execution, $p \in \Pi$ and $i \in \mathbb{N}$.
          The state of $p$ in $t$ after $i$ transitions is $q_p^t[i] \triangleq
          \langle \mathbf{card}(\{ l < i \mid t[l] = next_p\})+1,
          \{\langle r, k \rangle \mid \exists l < i:
          t[l] = deliver(r,k,p)\} \rangle)$
      \end{definition}

      Notice that such executions do not allow process to "jump" from say
      round $5$ to round $9$ without passing by the rounds in-between. The reason
      is that the Heard-Of model does not give processes access to the decision
      to change rounds: processes specify only which messages to send depending on
      the state, and what is the next state depending on the current state and
      the received messages.

      Also, the only information considered here is the round number and the
      received messages. This definition of execution disregards
      the message contents and the internal states of processes, as
      they are irrelevant to the implementation of Heard-Of predicates.

      Recall that strategies constrain when processes can change
      round. Thus, the executions that conform to a strategy change
      rounds only when allowed by it, and do it infinitely often if possible.

      \begin{definition}[Executions of a Strategy]
        \label{def:execsStrat}
          Let $f$ be a strategy and $t$ an execution.
          $t$ is an \textbf{execution of} $f \triangleq$
          $t$ satisfies:
          \begin{itemize}
              \item \textbf{(All nexts allowed)}
                  $\forall i \in \mathbb{N}, \forall p \in \Pi:
                  (t[i] = next_p  \implies q_p^t[i] \in f)$
              \item \textbf{(Fairness)}
                  $\forall p \in \Pi:
                  \mathbf{card}(\{i \in \mathbb{N} \mid t[i] = next_p\}) < \aleph_0
                  \implies \mathbf{card}(\{i \in \mathbb{N} \mid q_p^t[i] \notin f\})
                  = \aleph_0$
          \end{itemize}

          For a delivered predicate \textit{PDel}, $execs_f(\textit{PDel})
          \triangleq
          \{t \in execs(\textit{PDel}) \mid \textit{t is an}$ $\textit{execution of f } \}$.
      \end{definition}

      The fairness property can approximately be expressed in LTL as
      $\forall p \in \Pi: \lozenge \square (q_p^t \in f)\Rightarrow \square \lozenge next_p$.
      Note however that executions are here defined as sequences of transitions,
      whereas LTL models are sequences of states.

      An important part of this definition  considers
      executions where processes cannot necessarily
      change round after each delivery. That is, in the case of
      "waiting for at most $F$ messages", an execution where
      more messages are delivered than $F$ at some round is still an execution
      of the strategy.
      This hypothesis captures the asynchrony of processes,
      which are not always scheduled right after deliveries. It is compensated
      by a weak fairness assumption: if a strategy forever allows
      the change of round, it must eventually happen.

      Going back to strategies, not all of them are equally valuable.
      In general, strategies that block forever at some round are less useful than
      strategies that don't -- they forbid termination in some cases.
      The validity of a strategy captures the absence of such
      an infinite wait.

      \begin{definition}[Validity]
          \newline
          An execution $t$ is \textbf{valid}
          \mbox{$\triangleq
                  \forall p \in \Pi:
                  \mathbf{card}(\{i \in \mathbb{N} \mid t[i] = next_p\}) = \aleph_0$}.

          Let \textit{PDel} a delivered predicate and $f$ a strategy.
          $f$ is a \textbf{valid strategy} for \textit{PDel}
          $\triangleq \forall t \in execs_f(PDel): t$ is a valid execution.
      \end{definition}

      Because in a valid execution no process is ever blocked at a given
      round, there are infinitely many rounds. Hence, the messages delivered
      before the changes of round uniquely define a heard-of collection.

      \begin{definition}[Heard-Of Collection of Executions and
                         Heard-Of Predicate of Strategies]
          Let $t$ be a valid execution.
          $h_t$ is the \textbf{heard-of collection of} $t \triangleq\\
          \forall r \in \mathbb{N}^*, \forall p \in \Pi : h_t(r,p) =
          \left\{ k \in \Pi ~\middle|~ \exists i \in \mathbb{N}:
          \left(
          \begin{array}{ll}
              & q_p^t[i].round = r\\
              \land & t[i] = next_p\\
              \land & \langle r,k \rangle \in q_p^t[i].mes\\
          \end{array}
          \right)
          \right\}$

          Let $PDel$ be a delivered predicate, and $f$ be a valid strategy for $PDel$.
          We write $PHO_f(\textit{PDel})$ for the heard-of predicate
          composed of the collections of the executions of $f$
          on \textit{PDel}: $PHO_f(\textit{PDel})
          \triangleq \{ h_t \mid t \in execs_f(\textit{PDel}) \}$.
      \end{definition}

      Lastly, the heard-of predicate of most interest is the strongest one
      that can be generated by a valid strategy on the delivered predicate. Here
      strongest means the one that implies all the other heard-of predicates that
      can be generated on the same delivered predicate.
      The intuition boils down to two ideas:
      \begin{itemize}
        \item The strongest predicate implies all heard-of predicates
          generated on the same $PDel$, and thus it characterizes them
          completely.
        \item When seeing predicates as sets, implication is the reverse inclusion.
          Hence the strongest predicate is the one included in all the others.
          Less collections means more constrained communication, which
          means a more powerful model.
      \end{itemize}
      This notion of strongest predicate is formalized through an order on
      strategies and their heard-of predicates.

      \begin{definition}[Domination]
          Let \textit{PDel} be a delivered predicate and let
          $f$ and $f'$ be two valid strategies for \textit{PDel}.
          $f$~\textbf{dominates} $f'$ for \textit{PDel},
          written $f' \prec_{\textit{PDel}} f$, $\triangleq
          PHO_{f'}(\textit{PDel}) \supseteq PHO_f(\textit{PDel})$.

          A greatest element for $\prec_{\textit{PDel}}$ is called
          a \textbf{dominating strategy} for \textit{PDel}. Given
          such a strategy $f$, the \textbf{dominating predicate}
          for \textit{PDel} is $PHO_f(\textit{PDel})$.
      \end{definition}

    \subsection{Examples}
    \label{subsec:examples}

        We now show the variety of models that can be constructed
        from basic building blocks. Our basic blocks are the model $PDel^{total}$
        with only the collection $c_{total}$ where all the messages are delivered,
        and the model $PDel^{crash}_{1,r}$ with at most one crash that can happen at round $r$.

        \begin{definition}[At most $1$ crash at round $r$]
            \label{crash-round}
            $\mathcal{P}^{crash}_{1,r} \triangleq\\ \left\{
            c \text{ a delivered collection}
            \middle| \exists \Sigma \subseteq \Pi:
            \begin{array}{ll}
                & |\Sigma| \geq n-1\\
                \land & \forall j \in \Pi
                    \left(
                    \begin{array}{lll}
                        & \forall r' \in [1,r[:
                            & c(r',j) = \Pi\\
                        \land & & c(r,j) \supseteq \Sigma\\
                        \land & \forall r' \geq r:
                            & c(r',j) = \Sigma\\
                    \end{array}
                    \right)\\
            \end{array}
            \right\}$.
        \end{definition}

        From this family of predicates, various predicates can be built.
        Table~\ref{tab:examples} show some of them, as well as the Heard-Of
        predicates computed for these predicates based on the results from
        Section~\ref{subsec:HOCarefr} and Section~\ref{subsec:domCarefr}.
        For example the predicate with at most one crash $\mathcal{P}^{crash}_{1}$
        If a crash happens, it happens at one specific round $r$. We can
        thus build $\mathcal{P}^{crash}_{1}$ from a disjunction
        for all values of $r$ of the predicate with at most one crash at round $r$;
        that is, by the union of $\mathcal{P}^{crash}_{1,r}$ for all $r$.

        \begin{table}[t]
          \scriptsize
        \begin{tabular}{|l|l|l|c|}
            \hline
            Description & Expression & HO & Proof\\
            \hline
            At most 1 crash &
                $\mathcal{P}^{crash}_{1} =
                \bigcup\limits_{i=1}^{\infty} \mathcal{P}^{crash}_{1,i}$
                & $\textit{HOProd}(\{T \subseteq \Pi \mid |T|\geq n-1\})$
                & \cite{ShimiOPODIS18}\\
            \hline
            At most $F$ crashes &
                $\mathcal{P}^{crash}_{F} =
                \combi\limits_{j=1}^F \mathcal{P}^{crash}_1$
                & $\textit{HOProd}(\{T \subseteq \Pi \mid |T|\geq n-F\})$
                & \cite{ShimiOPODIS18}\\
            \hline
            \begin{tabular}{l}
                At most 1 crash,\\
                which will restart
            \end{tabular}&
                $\mathcal{P}^{recover}_{1} =
                \mathcal{P}^{crash}_1 \leadsto \mathcal{P}^{total}$
                & $\textit{HOProd}(\{T \subseteq \Pi \mid |T|\geq n-1\})$
                & Thm~\ref{oblivOpsHO}\\
            \hline
            \begin{tabular}{l}
                At most $F$ crashes,\\
                which will restart
            \end{tabular}&
                $\mathcal{P}^{recover}_{F} =
                \combi\limits_{j=1}^F \mathcal{P}^{recover}_{1}$
                & $\textit{HOProd}(\{T \subseteq \Pi \mid |T|\geq n-F\})$
                & Thm~\ref{oblivOpsHO}\\
            \hline
            \begin{tabular}{l}
                At most $1$ crash,\\
                which can restart
            \end{tabular}&
                $\begin{array}{l}
                   \mathcal{P}^{canrecover}_{1}\\
                   = \mathcal{P}^{recover}_{1} \cup \mathcal{P}^{crash}_1
                \end{array}$
                & $\textit{HOProd}(\{T \subseteq \Pi \mid |T|\geq n-1\})$
                & Thm~\ref{oblivOpsHO}\\
            \hline
            \begin{tabular}{l}
                At most $F$ crashes,\\
                which can restart
            \end{tabular}&
                $\begin{array}{l}
                \mathcal{P}^{canrecover}_{F}\\
                = \combi\limits_{j=1}^F \mathcal{P}^{canrecover}_{1}\\
                \end{array}$
                & $\textit{HOProd}(\{T \subseteq \Pi \mid |T|\geq n-F\})$
                & Thm~\ref{oblivOpsHO}\\
            \hline
            \begin{tabular}{l}
              No bound on crashes\\
              and restart, with only\\
              $1$ crash at a time
            \end{tabular}&
                $\mathcal{P}^{recovery}_1 =
                (\mathcal{P}^{crash}_1)^{\omega}$
                & $\textit{HOProd}(\{T \subseteq \Pi \mid |T|\geq n-1\})$
                & Thm~\ref{oblivOpsHO}\\
            \hline
            \begin{tabular}{l}
              No bound on crashes\\
              and restart, with max\\
              $F$ crashes at a time
            \end{tabular}&
                $\mathcal{P}^{recovery}_F =
                \combi\limits_{j=1}^F \mathcal{P}^{recovery}_1$
                & $\textit{HOProd}(\{T \subseteq \Pi \mid |T|\geq n-F\})$
                & Thm~\ref{oblivOpsHO}\\
            \hline
            \begin{tabular}{l}
                At most $1$ crash,\\
                after round $r$\\
            \end{tabular}&
                $\mathcal{P}^{crash}_{1,\geq r} =
                    \bigcup\limits_{i=r}^{\infty} \mathcal{P}^{crash}_{1,i}$
                & \begin{tabular}{l}
                    $\subseteq
                        \textit{HOProd}(\{T \subseteq \Pi \mid |T|\geq n-1\})$\\
                \end{tabular}
                & \begin{tabular}{l}
                    Thm~\ref{upperBoundHO}\\
                \end{tabular}\\
            \hline
            \begin{tabular}{l}
                At most $F$ crashes,\\
                after round $r$\\
            \end{tabular}&
                $\mathcal{P}^{crash}_{F,\geq r} =
                    \bigcup\limits_{i=r}^{\infty} \mathcal{P}^{crash}_{F,i}$
                & \begin{tabular}{l}
                    $\subseteq
                        \textit{HOProd}(\{T \subseteq \Pi \mid |T|\geq n-F\})$\\
                \end{tabular}
                & \begin{tabular}{l}
                    Thm~\ref{upperBoundHO}\\
                \end{tabular}\\
            \hline
            \begin{tabular}{l}
                At most $F$ crashes\\
                with no more than\\
                one per round\\
            \end{tabular}&
                $\begin{array}{l}
                   \mathcal{P}^{crash\neq}_{F}\\
                   = \bigcup\limits_{i_1 \neq i_2 \ldots\neq  i_F}
                   \combi\limits_{j=1}^{F} \mathcal{P}^{crash}_{1,i_j}\\
                \end{array}$
                & \begin{tabular}{l}
                    $\subseteq
                        \textit{HOProd}(\{T \subseteq \Pi \mid |T|\geq n-F\})$\\
                \end{tabular}
                & \begin{tabular}{l}
                    Thm~\ref{upperBoundHO}\\
                \end{tabular}\\
            \hline
        \end{tabular}
        \caption{A list of delivered predicate built using our operations,
        and their corresponding heard-of predicate. The $\textit{HOProduct}$ operator
        is defined in Definition~\ref{HOProd}.}
        \label{tab:examples}
        \end{table}

    \subsection{Families of strategies}
    \label{subsec:typesStrat}

    Strategies as defined above are predicates on states. This makes them incredibly
    expressive; on the other hand, this expressivity creates difficulty in reasoning
    about them.
    To address this problem, we define families of strategies. Intuitively,
    strategies in a same family depend on a specific part
    of the state -- for example the messages of the current round.
    Equality of these parts of the state defines an equivalence relation;
    the strategies of a family are strategies on the equivalence classes of
    this relation.

    \begin{definition}[Families of strategies]
      Let $\approx : Q \times Q \to bool$. The family of strategies defined
      by $\approx$, $family(\approx) \triangleq \{ f \text{ a strategy} \mid \forall
      q_1,q_2 \in \Pi: q_1 \approx q_2 \implies  (q_1 \in f \iff q_2 \in f)\}$
    \end{definition}

\section{Oblivious Strategies}
\label{sec:carefr}

    The simplest non-trivial strategies use only information from the messages of the
    current round. These strategies that do not remember messages from
    previous rounds, do not use messages in advance from future rounds, and
    do not use the round number itself.
    These strategies are called oblivious. They are simple, the Heard-Of predicates
    they implement are relatively easy to compute, and they require little
    computing power and memory to implement.
    Moreover, many examples above are dominated by such a strategy.
    Of course, there is a price to pay: oblivious strategies
    tend to be coarser than general ones.

    \subsection{Minimal Oblivious Strategy}
    \label{subsec:remCarefr}

        An oblivious strategy is defined by the different subsets of $\Pi$ from
        which it has to receive a message before allowing a change of round.

        \begin{definition}[Oblivious Strategy]
            Let $obliv$ be the function such that $\forall q \in Q: obliv(q) =
            \{k \in \Pi \mid \langle q.round, k \rangle \in q.mes\}$.
            Let $\approx_{obliv}$ the equivalence relation defined by
            $q_1 \approx_{obliv} q_2 \triangleq obliv(q_1) = obliv(q_2)$.
            The family of oblivious strategies is $family(\approx_{obliv})$.
            For $f$ an oblivious strategy, let $\textit{Nexts}_{f} \triangleq
            \{obliv(q) \mid q \in f \}$.
            It uniquely defines $f$.
        \end{definition}

        We will focus on a specific strategy, that dominates the oblivious
        strategies for a predicate. This follows from the fact that it waits
        less than any other valid oblivious strategy for this predicate.

        \begin{definition}[Minimal Oblivious Strategy]
            \label{minCarefr}
            Let $PDel$ be a delivered predicate. The \textbf{minimal oblivious strategy}
            for $PDel$ is $f_{min} \triangleq\\
            \{q \mid \exists c \in PDel, \exists p \in \Pi,
            \exists r  > 0: obliv(q) = c(r,p) \}$.
        \end{definition}

        \begin{lemma}[Domination of Minimal Oblivious Strategy]
            \label{domMinObliv}
            Let $PDel$ be a PDel and $f_{min}$ be its minimal oblivious strategy.
            Then $f_{min}$ is a dominating oblivious strategy for $PDel$.
        \end{lemma}

        \begin{proof}[Proof idea]
            $f_{min}$ is valid, because for every possible set of received
            messages in a collection of $PDel$, it accepts the corresponding
            oblivious state by definition of minimal oblivious strategy.
            It is dominating among oblivious strategies because any other
            valid oblivious strategy must allow the change of round when
            $f_{min}$ does it: it contains $f_{min}$.
            If an oblivious strategy does not contain $f_{min}$,
            then there is a collection of $PDel$ in which at a given round,
            a certain process might receive exactly the messages for the
            oblivious state accepted by $f_{min}$ and not by $f$. This
            entails that $f$ is not valid.
        \end{proof}

    \subsection{Operations Maintain Minimal Oblivious Strategy}
    \label{subsec:domOpsCarefr}

        As teased above, minimal oblivious strategies behave nicely under
        the proposed operations. That is, they give minimal oblivious strategies
        of resulting delivered predicates.
        One specificity of minimal oblivious strategies is that
        there is no need for the succession
        operation on strategies, nor for the repetition.
        An oblivious strategy has no
        knowledge about anything but the messages of the current round, and not
        even its round number, so it is impossible to distinguish a union from
        a succession, or a repetition from the initial predicate itself.

        \begin{theorem}[Minimal Oblivious Strategy for Union and Succession]
            \label{unionCarefr}
            Let $PDel_1, PDel_2$ be two delivered predicates, $f_1$ and $f_2$ the minimal
            oblivious strategies for, respectively, $PDel_1$ and $PDel_2$.
            Then $f_1 \cup f_2$ is the minimal oblivious strategy for
            $PDel_1 \cup PDel_2$ and $PDel_1 \leadsto PDel_2$.
        \end{theorem}

        \begin{proof}[Proof idea]
            Structurally, all proofs in this section consist in showing
            equality between the strategies resulting from the operations
            and the minimal oblivious strategy for the delivered predicate.

            For a union, the messages that can be received at each round are
            the messages that can be received at each round in the first
            predicate or in the second. This is also true for succession.
            Given that $f_1$ and $f_2$ are the minimal oblivious strategies
            of $PDel_1$ and $PDel_2$,
            they accept exactly the states with one of these sets of current
            messages. And thus $f_1 \cup f_2$ is the minimal oblivious strategy for
            $PDel_1 \cup PDel_2$ and $PDel_1 \leadsto PDel_2$.
        \end{proof}

        \begin{theorem}[Minimal Oblivious Strategy for Repetition]
            \label{repetCarefr}
            Let $PDel$ be a delivered predicate, and
            $f$ be its minimal oblivious strategy.
            Then $f$ is the minimal oblivious strategy for
            $PDel^{\omega}$.
        \end{theorem}

        \begin{proof}[Proof idea]
            The intuition is the same as for union and succession.
            Since repetition involves only one PDel, the sets of received
            messages do not change and $f$ is the minimal oblivious strategy.
        \end{proof}

        For combination, a special symmetry hypothesis is needed.

        \begin{definition}[Totally Symmetric PDel]
            Let $PDel$ be a delivered predicate. $PDel$ is \textbf{totally symmetric}
            $\triangleq \forall c \in PDel, \forall r > 0, \forall p \in
            \Pi, \forall r' > 0, \forall q \in \Pi, \exists c' \in PDel:
            c(r,p) = c'(r',q)$
        \end{definition}

        Combination is different because combining collections is done round by round.
        As oblivious strategies do not depend on the
        round, the combination of oblivious strategies creates the same
        combination of received messages for each round. We thus need these
        combinations to be independent of the round -- to be possible at
        each round -- to reconcile those two elements.

        \begin{theorem}[Minimal Oblivious Strategy for Combination]
            \label{combiCarefr}\\
            Let $PDel_1, PDel_2$ be two totally symmetric delivered predicates, $f_1$ and $f_2$ the minimal
            oblivious strategies for, respectively, $PDel_1$ and $PDel_2$.
            Then $f_1 \combi f_2$ is the minimal oblivious strategy for
            $PDel_1 \combi PDel_2$.
        \end{theorem}

        \begin{proof}[Proof idea]
            The oblivious states of $PDel_1 \combi PDel_2$ are the combination
            of an oblivious state of $PDel_1$ and of one of $PDel_2$ at the
            same round, for the same process. Thanks to total symmetry,
            this translates into the intersection of any oblivious state of
            $PDel_1$ with any oblivious state of $PDel_2$.
            Since $f_1$ and $f_2$ are the minimal oblivious strategy, they both
            accept exactly the oblivious states of $PDel_1$ and $PDel_2$
            respectively. Thus, $f_1 \combi f_2$ accept all combinations
            of oblivious states of $PDel_1$ and $PDel_2$, and thus
            is the minimal oblivious strategy of $PDel_1 \combi PDel_2$.
        \end{proof}

    \subsection{Computing Heard-Of Predicates}
    \label{subsec:HOCarefr}

        The computation of the heard-of predicate generated
        by an oblivious strategy is easy thanks to a characteristic of this HO:
        it is a product of sets of possible messages.

        \begin{definition}[Heard-Of Product]
            \label{HOProd}
            Let $S \subseteq \mathcal{P}(\Pi)$.
            The \textbf{heard-of product generated by S},
            $HOProd(S) \triangleq \{h \mid \forall p \in \Pi,
            \forall r > 0: h(r,p) \in S \}$.
        \end{definition}

        \begin{lemma}[Heard-Of Predicate of an Oblivious Strategy]
            \label{carefrHO}
            Let $PDel$ be a delivered predicate containing $c_{tot}$ and
            let $f$ be a valid oblivious strategy for $PDel$. Then
            $\textit{PHO}_f(PDel) = \textit{HOProd}(\textit{Nexts}_f)$.
        \end{lemma}

        \begin{proof}
            Proved in ~\cite[Theorem 20, Section 4.1]{ShimiOPODIS18}.
        \end{proof}

        Thanks to this characterization, the heard-of predicate
        generated by the minimal strategies for the operations is computed in terms
        of the heard-of predicate generated by the original minimal strategies.

        \begin{theorem}[Heard-Of Predicate of Minimal Oblivious Strategies]
            \label{oblivOpsHO}
            Let $PDel, PDel_1, PDel_2$ be delivered predicates containing $c_{tot}$.
            Let $f, f_1, f_2$ be their respective minimal oblivious strategies.
            Then:
            \begin{itemize}
                \item $PHO_{f_1 \cup f_2}(PDel_1 \cup PDel_2) =
                    PHO_{f_1 \cup f_2}(PDel_1 \leadsto PDel_2)\\
                    = \textit{HOProd}(Nexts_{f_1} \cup Nexts_{f_2})$.
                \item If $PDel_1$ or $PDel_2$ are totally symmetric,
                    $PHO_{f_1 \combi f_2}(PDel_1 \combi PDel_2)=
                    \textit{HOProd}(\{ n_1 \cap n_2 \mid
                    n_1 \in Nexts_{f_1} \land n_2 \in Nexts_{f_2}\})$.
                \item $PHO_f(PDel^{\omega}) = PHO_f(PDel)$.
            \end{itemize}
        \end{theorem}

        \begin{proof}[Proof idea]
            We apply Lemma~\ref{carefrHO}.
            The containment of $c_{tot}$ was shown in the proof of
            Theorem~\ref{invarSuffCfree}.
            As for the equality of the oblivious states, it follows from
            the intuition in the proofs of the minimal oblivious strategy in
            the previous section.
        \end{proof}

    \subsection{Domination by an Oblivious Strategy}
    \label{subsec:domCarefr}

      From the previous sections, we can compute the Heard-Of predicate of
      the dominating oblivious strategies for our examples. We first need to
      give the minimal oblivious strategy for our building blocks $PDel^{crash}_1$ and $PDel^{total}$.

      \begin{definition}[Waiting for $n-F$ messages]
        The strategy to wait for $n-F$ messages is:
          $f^{n,F} \triangleq
          \{ q \in Q \mid |obliv(q)| \geq n-F \}$
      \end{definition}

      For all $F < n$, $f^{n,F}$ is the minimal oblivious strategy for
      $PDel^{crash}_F$ (shown by Shimi et al.~\cite[Thm. 17]{ShimiOPODIS18}).
      For $PDel^{total}$, since every process receives all the messages all the time,
      the strategy waits for all the messages ($f^{n,0}$).

      Using these strategies, we deduce the heard-of predicates of
      dominating oblivious strategies for our examples.

      \begin{itemize}\sloppy
        \item For $PDel^{recover}_1 \triangleq PDel^{crash}_1 \leadsto PDel^{total}$,
          the minimal oblivious strategy $f^{recover}_1 = f^{n,1} \cup f^{n,0}
          = f^{n,1}$.
          This entails that\\
          $PHO_{f^{recover}_1} =
          \textit{HOProd}(\{T \subseteq \Pi \mid |T|\geq n-1\})$.
        \item For $PDel^{canrecover}_1 \triangleq PDel^{recover}_1 \cup PDel^{crash}_1$,
          the minimal oblivious strategy $f^{canrecover}_1 = f^{recover}_1 \cup f^{n,1}
          = f^{n,1}$.
          This entails that\\
          $PHO_{f^{canrecover}_1} =
          \textit{HOProd}(\{T \subseteq \Pi \mid |T|\geq n-1\})$.

        \item For $PDel^{crash}_1 \combi PDel^{canrecover}_1$
          the minimal oblivious strategy\\
          $f = f^{n,1} \combi f^{canrecover}_1
          = f^{n,1} \combi f^{n,1} = f^{n,2}$.
          This entails that\\
          $PHO_f = \textit{HOProd}(\{T \subseteq \Pi \mid |T|\geq n-2\})$.
      \end{itemize}

	The computed predicate is the predicate of the dominating \textit{oblivious} strategy.
  But the dominating strategy might not be oblivious,
  and this predicate might be too weak.
	The following result shows that $PDel^{crash}_1$ and $PDel^{total}$
  satisfy conditions that imply their domination by an oblivious strategy.
  Since these conditions are invariant by our operations,
  all PDel constructed with these building blocks are dominated by an oblivious strategy.



      \begin{theorem}[Domination by Oblivious for Operations]
          \label{invarSuffCfree}\\
          Let $PDel, PDel_1, PDel_2$ be delivered predicates that satisfy:
          \begin{itemize}
            \item \textbf{(Total collection)}
              They contains the total collection $c_{tot}$,
            \item \textbf{(Symmetry up to a round)}
              $\forall c$ a collection in the predicate, $\forall p \in \Pi,
              \forall r > 0, \forall r' > 0, \exists c'$ a collection
              in the predicate:
              $c'[1,r'-1] = c_{tot}[1,r'-1] \land \forall q \in \Pi:
              c'(r',q)=c(r,p)$
          \end{itemize}
          Then $PDel_1 \cup PDel_2$, $PDel_1 \combi PDel_2$, $PDel_1 \leadsto
          PDel_2$, $PDel^{\omega}$ satisfy the same two conditions and
          are dominated by oblivious strategies.
      \end{theorem}

      Both $\mathcal{P}^{crash}_1$ from Table~\ref{tab:examples} and
      $\mathcal{P}^{total} = \{c_{tot}\}$ satisfy this condition.
      So do all the first 8 examples from Table~\ref{tab:examples}, since they
      are built from these two.

\section{Conservative Strategies}
\label{sec:cons}

    We now broaden our family of considered strategies, by allowing
    them to consider past and present rounds, as well as the round number
    itself. This is a generalization of oblivious strategies, that
    tradeoff simplicity for expressivity, while retaining a
    nice structure.
    Even better, we show that both our building blocks and all the predicates
    built from them are dominated by such a strategy. For the examples then,
    no expressivity is lost.

    \subsection{Minimal Conservative Strategy}
    \label{subsec:remindersReac}

        \begin{definition}[Conservative Strategy]
            Let $cons$ be the function such that $\forall q \in Q,\ cons(q) \triangleq
            \langle q.round, \{ \langle r, k \rangle \in q.mes \mid r \le q.round\}\rangle$.
            Let $\approx_{cons}$ the equivalence relation defined by
            $q_1 \approx_{cons} q_2 \triangleq cons(q_1) = cons(q_2)$.
            The family of conservative strategies is \mbox{$family(\approx_{cons})$}.
            We write $\textit{Nexts}^R_f \triangleq \{cons(q) \mid q \in f\}$ for the set
            of conservative states in $f$. This uniquely defines $f$.
        \end{definition}

        In analogy with the case of oblivious strategies, we can define
        a minimal conservative strategy of
        $PDel$, and it is a strategy dominating all conservative strategies for
        this delivered predicate.

        \begin{definition}[Minimal Conservative Strategy]
            \label{minReac}
            Let $PDel$ be a delivered predicate.
            The \textbf{minimal conservative strategy}
            for $PDel$ is $f_{min} \triangleq$ the conservative strategy
            such that $f =
            \{q \in Q \mid \exists c \in PDel, \exists p \in \Pi,
            \forall r \leq q.round: q(r) = c(r,p) \}$.
        \end{definition}

%

        \begin{lemma}[Domination of Minimal Conservative Strategy]
            \label{minDom}
            Let $PDel$ be a delivered predicate and $f_{min}$
            be its minimal conservative strategy.
            Then $f_{min}$ dominates the conservative strategies for $PDel$.
        \end{lemma}

        \begin{proof}[Proof idea]
            Analogous to the case of minimal oblivious strategies:
            it is valid because it allows to change round for each
            possible conservative state (the round and the messages received
            for this round and before) of collections in $PDel$.
            And since any other valid conservative strategy $f$ must
            accept these states (or it would block forever in some
            execution of a collection of $PDel$), we have that
            $f$ contains $f_{min}$ and thus that $f_{min}$ dominates
            $f$.
        \end{proof}

    \subsection{Operations Maintain Minimal Conservative Strategies}
    \label{subsec:domOpsReac}

        Like oblivious strategies, minimal conservative strategies
        give minimal conservative strategies of resulting delivered predicates.

        \begin{theorem}[Minimal Conservative Strategy for Union]
            \label{unionReac}\\
            Let $PDel_1, PDel_2$ be two delivered predicates, $f_1$ and $f_2$ the minimal
            conservative strategies for, respectively, $PDel_1$ and $PDel_2$.
            Then $f_1 \cup f_2$ is the minimal conservative strategy for
            $PDel_1 \cup PDel_2$.
        \end{theorem}

        \begin{proof}[Proof idea]
          A prefix of a collection in $PDel_1 \cup PDel_2$ comes from
          either $PDel_1$ or $PDel_2$, and thus is accepted by $f_1$
          or $f_2$. And any state accepted by $f_1 \cup f_2$ corresponds
          to some prefix of $PDel_1$ or $PDel_2$.
        \end{proof}

        For the other three operations, slightly more structure is needed on
        the predicates. More precisely, they have to be independent of
        the processes. Any prefix of a process $p$ in a collection
        of the predicate is also the prefix of any other process $q$
        in a possibly different collection of the same PDel. Hence,
        the behaviors (fault, crashes, loss) are not targeting specific
        processes.
        This restriction fits the intuition
        behind many common fault models.

        \begin{definition}[Symmetric PDel]
            Let $PDel$ be a delivered predicate. $PDel$ is \textbf{symmetric}
            $\triangleq \forall c \in PDel, \forall p \in \Pi, \forall r > 0,
            \forall q \in \Pi, \exists c' \in PDel, \forall r' \leq r:
            c'(r',q) = c(r',p)$
        \end{definition}

        \begin{theorem}[Minimal Conservative Strategy for Combination]
            \label{combiReac}\\
            Let $PDel_1, PDel_2$ be two symmetric delivered predicates,
            $f_1$ and $f_2$ the minimal
            conservative strategies for, respectively, $PDel_1$ and $PDel_2$.
            Then $f_1 \combi f_2$ is the minimal conservative strategy for
            $PDel_1 \combi PDel_2$.
        \end{theorem}

        \begin{proof}[Proof idea]
            Since $f_1$ and $f_2$ are the minimal conservative strategies
            of $PDel_1$ and
            $PDel_2$, $Nexts^R{f_1}$ is the set of the conservative states of
            prefixes of $PDel_1$ and $Nexts^R_{f_2}$ is the set of the conservative
            states of prefixes of $PDel_2$.
            Also, the states accepted by $f_1 \combi f_2$ are the combination
            of the states accepted by $f_1$ and the states accepted by $f_2$.
            And the prefixes of $PDel_1 \combi PDel_2$ are the prefixes of
            $PDel_1$ combined with the prefixes of $PDel_2$ \textbf{for
            the same process}. Thanks to symmetry,
            we can take a prefix of $PDel_2$ and any process, and find
            a collection such that the process has that prefix.
            Therefore the combined prefixes for the same process are the
            same as the combined prefixes of $PDel_1$ and $PDel_2$.
            Thus, $Nexts^R_{f_1 \combi f_2}$ is the set of conservative states
            of prefixes of $PDel_1 \combi PDel_2$,
            and $f_1 \combi f_2$ is its minimal conservative strategy.
        \end{proof}

        \begin{theorem}[Minimal Conservative Strategy for Succession]
            \label{succReac}\\
            Let $PDel_1, PDel_2$ be two symmetric delivered predicates,
            $f_1$ and $f_2$ the minimal
            conservative strategies for, respectively, $PDel_1$ and $PDel_2$.
            Then $f_1 \leadsto f_2$ is the minimal conservative strategy for
            $PDel_1 \leadsto PDel_2$.
        \end{theorem}

        \begin{proof}[Proof idea]
            Since $f_1$ and $f_2$ are the minimal conservative strategies
            of $PDel_1$ and
            $PDel_2$, $Nexts^R{f_1}$ is the set of the conservative states of
            prefixes of $PDel_1$ and $Nexts^R_{f_2}$ is the set of the conservative
            states of prefixes of $PDel_2$.
            Also, the states accepted by $f_1 \leadsto f_2$ are the succession
            of the states accepted by $f_1$ and the states accepted by $f_2$.
            And the prefixes of $PDel_1 \leadsto PDel_2$ are the successions
            of prefixes of $PDel_1$ and prefixes of $PDel_2$ \textbf{for
            the same process}. But thanks to symmetry,
            we can take a prefix of $PDel_2$ and any process, and find
            a collection such that the process has that prefix.

            Therefore the succession of prefixes for the same process are the
            same as the succession of prefixes of $PDel_1$ and $PDel_2$.
            Thus, $Nexts^R_{f_1 \leadsto f_2}$ is the set of conservative
            states of prefixes of $PDel_1 \leadsto PDel_2$,
            and is therefore its minimal conservative strategy.
        \end{proof}

        \begin{theorem}[Minimal Conservative Strategy for Repetition]
            \label{repetReac}\\
            Let $PDel$ be a symmetric delivered predicate,
            and $f$ be its minimal conservative strategy.
            Then $f^{\omega}$ is the minimal conservative strategy for
            $PDel^{\omega}$.
        \end{theorem}

        \begin{proof}[Proof idea]
            The idea is the same as in the succession.
        \end{proof}

    \subsection{Computing Heard-Of Predicates}
    \label{subsec:HOReac}

        Here we split from the analogy with oblivious strategies:
        the heard-of predicate of conservative strategies is hard to compute,
        as it dependss in intricate ways on the delivered predicate itself.

        Yet it is still possible to compute interesting information on
        this HO: upper bounds. These are overapproximations of the
        actual HO, but they can serve for formal verification of LTL properties.
        Indeed, the executions of an algorithm for the actual HO are contained
        in the executions of the algorithm for any overapproximation of the
        HO, and LTL properties must be true for all executions of the algorithm.
        So proving the property on an overapproximation also proves it on
        the actual HO.

        \begin{theorem}[Upper Bounds on HO of Minimal Conservative Strategies]
            \label{upperBoundHO}
            Let $PDel, PDel_1, PDel_2$ be delivered predicates containing $c_{tot}$.\\
            Let $f^{cons}, f_1^{cons}, f_2^{cons}$ be their respective minimal
            conservative strategies,\\
            and $f^{obliv}, f_1^{obliv}, f_2^{obliv}$
            be their respective minimal oblivious strategies.
            Then:
            \begin{itemize}
                \item $PHO_{f_1^{cons} \cup f_2^{cons}}(PDel_1 \cup PDel_2)
                    \subseteq \textit{HOProd}(\textit{Nexts}_{f_1^{obliv}} \cup
                    \textit{Nexts}_{f_2^{obliv}})$.
                \item $PHO_{f_1^{cons} \leadsto f_2^{cons}}(PDel_1 \leadsto PDel_2)
                    \subseteq \textit{HOProd}(\textit{Nexts}_{f_1^{obliv}} \cup
                    \textit{Nexts}_{f_2^{obliv}})$.
                \item $PHO_{f_1^{cons} \combi f_2^{cons}}(PDel_1 \combi PDel_2)
                    \subseteq \textit{HOProd}( \{ n_1 \cap n_2 \mid
                    n_1 \in \textit{Nexts}_{f_1^{obliv}} \land
                    n_2 \in \textit{Nexts}_{f_2^{obliv}}\})$.
                \item $PHO_{(f^{cons})^\omega}(PDel^{\omega})
                    \subseteq \textit{HOProd}(\textit{Nexts}_{f^{obliv}})$.
            \end{itemize}
        \end{theorem}

        \begin{proof}[Proof idea]
        These bounds follow from the fact that an oblivious strategy,
        is a conservative strategy, and
        thus the minimal conservative strategy dominates the minimal
        oblivious strategy.
        \end{proof}

\section{Conclusion}
\label{conclu}

    To summarize, we propose operations on delivered predicates that
    allow the construction of complex predicates from simpler ones. The
    corresponding operations on strategies behave nicely regarding dominating
    strategies, for the conservative and oblivious strategies. This entails
    bounds and characterizations of the dominating heard-of predicate
    for the constructions.

    What needs to be done next comes in two kinds: first, the logical
    continuation is to look for constraints on delivered predicates for
    which we can compute the dominating heard-of predicate of conservative
    strategies. More ambitiously, we will study strategies looking in the
    future, i.e. strategies that can take into account messages from
    processes that have already reached a strictly higher round than the
    recipient. These strategies are useful for inherently
    asymmetric delivered predicates. For example, message loss
    is asymmetric, in the sense that we cannot
    force processes to receive the same set of messages.

\paragraph{Funding} This work was supported by project PARDI ANR-16-CE25-0006.



\bibliography{references}


\newpage
\appendix

\section{Tools}
\label{app:tools}

        \subsection{Timing Functions}

        A timing function of an execution captures the round at which a message
        is delivered: for a message sent in round $r'$ by $k$ to $j$,
        $time(r',k,j)$ is the round at which this message is delivered to $j$.
        Note that $time(r',k,j) = 0$ if and only if no message sent
        from $k$ to $j$ at round $r'$ is delivered in this execution.

        \begin{definition}[Timing Function]
            A \textbf{timing function} is a function
            $\mathbb{N}^* \times \Pi \times \Pi \mapsto \mathbb{N}$.

            For $t$ an execution, the timing function of $t$, $time_t
            \triangleq$ the timing function such that
            $\forall r > 0,\forall r' > 0, \forall k,j \in \Pi:
            time_t(r',k,j) = r \iff (\exists i \geq 0: t[i] = deliver(r',k,j)
            \land q^t_j[i].round = r)$.
        \end{definition}

        The standard execution reorders deliveries and changes of round
        such that all the deliveries for a given round happen before
        the changes of round for all processes.

        \begin{definition}[Standard Execution of a timing function]
            Let $time$ be a timing function and $ord$ be any function
            taking a set and returning an ordered sequence of its elements.
            The specific ordering doesn't matter.

            The \textbf{standard execution with timing $time$} is $st_{time}
            \triangleq \prod\limits_{r \in \mathbb{N}^*} dels_r.nexts$,
            where $dels_r = ord(\{deliver(r',k,j) \mid r' > 0 \land k,j \in \Pi
            \land time(r',k,j) = r\})$ and $nexts = ord(\{next_p \mid p \in \Pi\})$.
        \end{definition}

        \begin{lemma}[Correctness of Standard Execution with Timing]
            \label{stCorrect}
            Let $time$ be a timing function. Then
            $(\forall r > 0, \forall k,j \in \Pi: time(r,k,j) = 0 \lor
            time(r,k,j) \geq r)
            \implies st_{time}$ is an execution.
        \end{lemma}

        \begin{proof}
            \begin{itemize}
                \item \textbf{(Delivered after sending)}
                    Let $r > 0$ and $k,j \in \Pi$. If $time(r,k,j) = 0$,
                    then the message is never delivered, and we don't have
                    to consider it. If not, then by hypothesis
                    $time(r,k,j) \geq r$. This means $\exists i \geq r:
                    deliver(r,k,j) \in dels_i$.

                    By construction of the standard execution,
                    there are $i-1$ occurrences of the sequence $nexts$
                    before the sequence $dels_i$. This means
                    there are $i-1 \geq r-1$ $next_k$ before,
                    which allows us to conclude.
                \item \textbf{(Delivered only once)}
                    Let $r > 0$ and $k,j \in \Pi$. If $\exists i \geq 0:
                    st_{time}[i] = deliver(r,k,j)$, then it is in
                    $dels_{time(r,k,j)}$.
                    We conclude that there is only one delivery of this
                    message.
                  \item \textbf{(Once stopped, forever stopped)}
                    The standard execution does not contain any $stop$.
            \end{itemize}
        \end{proof}

        \begin{lemma}[Heard-Of Collection of Timing Function]
            \label{hoTiming}
            Let $t$ be a valid execution, and $time$ be its timing function.
            Then $\forall r > 0, \forall p \in \Pi:
            h_t(r,p) = \{ q \in \Pi \mid time(r,q,p) \in [1,r]\}$.
        \end{lemma}

        \begin{proof}
            Let $i \geq 0$ such that $\exists p \in \Pi: t[i]=next_p$.
            Let $r = q^t_p[i].round$.
            We show both side of $h_t(r,p) = \{ q \in \Pi \mid time(r,q,p) \in [1,r]\}$.
            \begin{itemize}
                \item Let $q \in h_t(r,p)$. Then it is delivered in a round $\leq r$,
                    and thus $time(r,q,p) \in [1,r]$.
                \item Let $q \in \Pi$ such that $time(r,q,p) \in [1,r]$. Then
                    by definition of $time$, the message sent by $q$ at round $r$
                    is delivered to $p$ in $t$ at most at round $r$. Thus,
                    it is in the messages from the current round when
                    going to round $r+1$, and $q \in h_t(r,p)$.
            \end{itemize}
        \end{proof}

\section{Proofs for Oblivious Strategies}
\label{app:obliv}

  \subsection{Minimal Oblivious Strategies}

  We use a necessary and sufficient condition for an oblivious strategy to be valid in
  the rest of the proofs.

        \begin{lemma}[Necessary and Sufficient Condition for Validity of a Oblivious Strategy]
            \label{carefrValid}
            Let $PDel$ be a delivered predicate and $f$ be an oblivious strategy.
            Then $f$ is valid for $PDel \iff
            f \supseteq \{q \mid \exists c \in PDel, \exists p \in \Pi,
            \exists r  > 0: obliv(q) = c(r,p) \}$.
        \end{lemma}

        \begin{proof}
            From the version in OPODIS 2018, $f$ has to satisfy
            $\forall c \in \textit{PDel}, \forall r > 0,
            \forall p \in \Pi: c(r,p) \in \textit{Nexts}_f$.

            We show the equivalence of this condition with our own,
            which allow us to conclude by transitivity of equivalence.
            \begin{itemize}
                \item $(\Longrightarrow)$ We assume our condition holds and
                    prove the one form OPODIS 2018.

                    Let $c \in PDel, r > 0$ and $p \in \Pi$: we want to show that
                    $c(r,p) \in \textit{Nexts}_f$. That is to say, that all states
                    whose present corresponds to this oblivious state
                    are accepted by $f$.

                    Let $q$ such that $obliv(q) = c(r,p)$.
                    We have the collection $c$, the round $r$ and the
                    process $p$ to apply our condition, and thus $q \in f$.

                    Hence, $c(r,p) \in \textit{Nexts}_f$.
                \item $(\Longleftarrow)$ We assume the condition from OPODIS 2018
                    holds and we prove ours.

                    Let $q$ such that $\exists c \in PDel, \exists p \in \Pi,
                    \exists r \leq q.round: obliv(q) = c(r,p)$.
                    By hypothesis, we have $c(r,p) \in \textit{Nexts}_f$.

                    We conclude that $q \in f$.
            \end{itemize}
        \end{proof}

        \begin{lemma*}[(\ref{domMinObliv}
                       Domination of Minimal Oblivious Strategy]
            Let $PDel$ be a PDel and $f_{min}$ be its minimal oblivious strategy.
            Then $f_{min}$ is a dominating oblivious strategy for $PDel$.
        \end{lemma*}

        \begin{proof}
            First, $f_{min}$ is valid for $PDel$ by application of
            Lemma~\ref{carefrValid}.
            Next, we take another oblivious strategy $f$, which is valid
            for $PDel$. Lemma~\ref{carefrValid} now gives us that
            $f_{min} \subseteq f$. Hence, when
            $f_{min}$ allow a change of round, so does $f$.
            This entails that all executions of $f_{min}$ for
            $PDel$ are also executions of $f$ for $PDel$,
            and thus that heard-of predicate generated by $f_{min}$ is contained
            in the one generated by $f$.
        \end{proof}

  \subsection{Operations Maintain Minimal Oblivious Strategies}

        \begin{theorem*}[(\ref{unionCarefr})
                        Minimal Oblivious Strategy for Union and Succession]
            Let $PDel_1, PDel_2$ be two delivered predicates, $f_1$ the minimal
            oblivious strategy for $PDel_1$, and $f_2$ the
            minimal oblivious strategy for $PDel_2$.
            Then $f_1 \cup f_2$ is the minimal oblivious strategy for
            $PDel_1 \cup PDel_2$ and $PDel_1 \leadsto PDel_2$.
        \end{theorem*}

        \begin{proof}
            We first show that the minimal oblivious strategies of
            $PDel_1 \cup PDel_2$ and $PDel_1 \leadsto PDel_2$ are equal.
            Hence, we prove
            $\{q \mid \exists c \in PDel_1 \cup PDel_2, \exists p \in \Pi,
            \exists r > 0: obliv(q) = c(r,p) \} =
            \{q \mid \exists c \in PDel_1 \leadsto PDel_2, \exists p \in \Pi,
            \exists r > 0: obliv(q) = c(r,p) \}$.
            \begin{itemize}
                \item $(\subseteq)$ Let $q$ such that $\exists c \in PDel_1
                    \cup PDel_2, \exists p \in \Pi, \exists r > 0: obliv(q)
                    = c(r,p)$.
                    \begin{itemize}
                        \item If $c \in PDel_1$, then we take $c_2 \in PDel_2$
                            $c' = c[1,r].c_2$. Since $c' \in
                            c \leadsto c_2$, we have $c' \in PDel_1 \leadsto PDel_2$.
                            And by definition of $c'$, $c'(r,p) = c(r,p)$.

                            We thus have $c', p$ and $r$ showing that $q$
                            is in the set on the right.
                        \item If $c \in PDel_2$, then
                            $c \in PDel_1 \leadsto PDel_2$
                            We thus have $c, p$ and $r$ showing that $q$
                            is in the set on the right.
                    \end{itemize}
                \item $(\supseteq)$ Let $q$ such that $\exists c \in PDel_1
                    \leadsto PDel_2, \exists p \in \Pi, \exists r > 0: obliv(q)
                    = c(r,p)$.
                    \begin{itemize}
                        \item If $c \in PDel_2$, then $c \in PDel_1 \cup
                            PDel_2$.
                            We thus have $c, p$ and $r$ showing that $q$
                            is in the set on the left.
                        \item If $c \notin PDel_2$, there exist
                            $c_1 \in PDel_1, c_2 \in PDel_2$ and
                            $r' > 0$ such that $c = c_1[1,r'].c_2$.
                            \begin{itemize}
                                \item If $r \leq r'$, then by definition of
                                    $c$, we have $c(r,p) = c_1(r,p)$.
                                    We thus have $c_1, p$ and $r$ showing that $q$
                                    is in the set on the left.
                                \item If $r > r'$, then $c(r,p) = c_2(r-r',p)$
                                    We thus have $c_2, p$ and $(r-r')$
                                    showing that $q$ is in the set on the left.
                            \end{itemize}
                    \end{itemize}
            \end{itemize}

            We show that $f_1 \cup f_2 =
            \{q \mid \exists c \in PDel_1 \cup PDel_2, \exists p \in \Pi,
            \exists r > 0: obliv(q) = c(r,p) \}$, which allows us to conclude
            by Definition~\ref{minCarefr}.
            \begin{itemize}
                \item Let $q \in f_1 \cup f_2$. We fix $q \in f_1$
                    (the case $q \in f_2$ is completely symmetric).

                    Then because $f_1$ is the minimal oblivious strategy of $PDel_1$,
                    by application of Lemma~\ref{carefrValid},
                    $\exists c_1 \in PDel_1,\exists p \in \Pi, \exists r > 0$
                    such that $c_1(r,p)= obliv(q)$.
                    $c_1 \in PDel_1 \subseteq PDel_1 \cup PDel_2$.
                    We thus have $c_1, p$ and $r$ showing that $q$
                    is in the minimal oblivious strategy for $PDel_1 \cup PDel_2$.
                \item Let $q$ such that $\exists c \in PDel_1 \cup PDel_2,
                    \exists p \in \Pi, \exists r > 0: c(r,p)= obliv(q)$.
                    By definition of union, $c$ must be in $PDel_1$ or
                    in $c \in PDel_2$; we fix $c \in PDel_1$
                    (the case $PDel_2$ is symmetric).

                    Then Definition~\ref{minCarefr} gives us that
                    $q$ is in the minimal oblivious strategy of $PDel_1$, that is $f_1$.
                    We conclude that $q \in f_1 \cup f_2$.
            \end{itemize}
        \end{proof}

        \begin{theorem*}[(\ref{repetCarefr})
                         Minimal Oblivious Strategy for Repetition]
            Let $PDel$ be a delivered predicate, and
            $f$ be its minimal oblivious strategy.
            Then $f$ is the minimal oblivious strategy for
            $PDel^{\omega}$.
        \end{theorem*}

        \begin{proof}
            We show that $f =
            \{q \mid \exists c \in PDel^{\omega}, \exists p \in \Pi,
            \exists r > 0: obliv(q) = c(r,p) \}$, which allows us to conclude
            by Definition~\ref{minCarefr}.
            \begin{itemize}
                \item $(\subseteq)$ Let $q \in f$. By minimality of $f$
                    for $PDel$,
                    $\exists c \in PDel, \exists p \in \Pi, \exists r > 0:
                    obliv(q) = c(r,p)$.

                    We take $c' \in PDel^{\omega}$ such that
                    $c_1 = c$ and $r_2 = r$; the other $c_i$ and
                    $r_i$ don't matter for the proof.
                    By definition of repetition, we get
                    $c'(r,p) = c(r,p) = obliv(q)$.

                    We have $c', p$ and $r$ showing that
                    $q$ is in the minimal oblivious strategy of
                    $PDel^{\omega}$.
                \item $(\supseteq)$ Let $q$ such that
                    $\exists c \in PDel^{\omega}, \exists p \in \Pi, \exists r > 0:
                    obliv(q) = c(r,p)$.
                    By definition of repetition, there are
                    $c_i \in PDel$ and $0< r_i < r_{i+1}$
                    such that $r \in [r_i+1,r_{i+1}]$ and
                    $c(r,p) = c_i(r-r_i,p)$.

                    We have found $c_i, p$ and $(r - r_i)$ showing that
                    $q$ is in the minimal oblivious strategy for $PDel$.
                    And since $f$ is the minimal oblivious strategy for $PDel$,
                    we get $q \in f$.
            \end{itemize}
        \end{proof}

        \begin{theorem*}[(\ref{combiCarefr})
                         Minimal Oblivious Strategy for Combination]
            Let $PDel_1, PDel_2$ be two totally symmetric delivered predicate,
            $f_1$ the minimal
            oblivious strategy for $PDel_1$, and $f_2$ the
            minimal oblivious strategy for $PDel_2$.
            Then $f_1 \combi f_2$ is the minimal oblivious strategy for
            $PDel_1 \combi PDel_2$.
        \end{theorem*}

        \begin{proof}
            We show that $f_1 \combi f_2 =
            \{q \mid \exists c \in PDel_1 \combi PDel_2, \exists p \in \Pi,
            \exists r > 0: obliv(q) = c(r,p) \}$, which allows us to apply
            Lemma~\ref{minCarefr}.

            \begin{itemize}
                \item
                    Let $q \in f_1 \combi f_1$. Then $\exists q_1 \in f_1,
                    \exists q_2 \in f_2$ such that $q = q_1 \combi q_2$.
                    This also means that $q_1.round = q_2.round = q.round$.

                    By minimality of $f_1$ and $f_2$,
                    $\exists c_1 \in PDel_1, \exists  p_1 \in \Pi, \exists r_1 > 0:
                    c_1(r_1,p_1) = obliv(q_1)$ and
                    $\exists c_2 \in PDel_2,\exists p_2 \in \Pi,\exists r_2 > 0:
                    c_2(r_2,p_2) = obliv(q_2)$.

                    Moreover, total symmetry of $PDel_2$ ensures that
                    $\exists c'_2 \in PDel_2:c'_2(r_1,p_1)
                    = c_2(r_2,p_2)$.

                    We take $c = c_1 \combi c'_2$.
                    $obliv(q) = obliv(q_1) \cap
                    obliv(q_2) = c_1(r_1,p_1) \cap c_2(r_2,p_2)
                    = c_1(r_1,p_1) \cap c'_2(r_1,p_1)
                    = c(r_1,p_1)$.

                    We have $c$, $p_1$ and $r_1$ showing that
                    $q$ is in the minimal oblivious strategy for $PDel_1 \combi PDel_2$.
                \item Let $q$ such that $\exists c \in PDel_1 \combi PDel_2,
                    \exists p \in \Pi, \exists r > 0: c(r,p)= obliv(q)$.
                    By definition of Combination, $\exists c_1 \in PDel_1,
                    \exists c_2 \in PDel_2: c = c_1 \combi c_2$.

                    We take $q_1$ such that $q_1.round = r,
                    obliv(q_1) = c_1(r,p)$ and
                    $\forall r' \neq r: q_1(r') = q(r')$; we also
                    take $q_2$ such that $q_2.round = r,
                    obliv(q_2) = c_2(r,p)$ and
                    $\forall r' \neq r: q_2(r') = q(r')$.

                    Then $q = q_1 \combi q_2$. And since $f_1$ and $f_2$
                    are the minimal oblivious strategies of $PDel_1$ and $PDel_2$ respectively,
                    we have $q_1 \in f_1$ and $q_2 \in f_2$.

                    We conclude that $q \in f_1 \combi f_2$.
            \end{itemize}
        \end{proof}

  \subsection{Computing Heard-Of Predicates}

        \begin{theorem*}[(\ref{oblivOpsHO})
                        Heard-Of Predicate of Minimal Oblivious Strategies]
            Let $PDel, PDel_1, PDel_2$ be delivered predicates containing $c_{tot}$.
            Let $HO, HO_1, HO_2$ be their respective $HO$,
            and let $f, f_1, f_2$ be their respective minimal oblivious strategies.
            Then:
            \begin{itemize}
                \item The HO generated by $f_1 \cup f_2$ on $PDel_1 \cup PDel_2$,
                    and on $PDel_1 \leadsto PDel_2$ is the HO product generated
                    by $Nexts_{f_1} \cup Nexts_{f_2}$.
                \item The HO generated by $f_1 \combi f_2$ on $PDel_1 \combi PDel_2$,
                    if either $PDel_1$ or $PDel_2$ is totally symmetric,
                    is the HO product generated by
                    $\{ n_1 \cap n_2 \mid
                    n_1 \in Nexts_{f_1} \land n_2 \in Nexts_{f_2}\}$.
                \item The HO generated by $f$ on $PDel^{\omega}$ is $HO$.
            \end{itemize}
        \end{theorem*}

        \begin{proof}
            Obviously, we want to apply Lemma~\ref{carefrHO}. Then
            we first need to show that our PDels contain $c_{tot}$.
            \begin{itemize}
                \item If $PDel_1$ and $PDel_2$ contain $c_{tot}$,
                    then $PDel_1 \cup PDel_2$ trivially contains it too.
                \item If $PDel_1$ and $PDel_2$ contain $c_{tot}$,
                    then $PDel_1 \combi PDel_2$ contains
                    $c_{tot} \combi c_{tot} = c_{tot}$.
                \item If $PDel_1$ and $PDel_2$ contain $c_{tot}$,
                    then $PDel_1 \leadsto PDel_2 \supseteq PDel_2$
                    contains it too.
                \item If $PDel$ contains $c_{tot}$,
                    we can recreate $c_{tot}$ by taking all $c_i = c_{tot}$
                    and whichever $r_i$. Thus, $PDel^{\omega}$ contains
                    $c_{tot}$.
            \end{itemize}

            Next, we need to show that the $\textit{Nexts}_f$ for the strategies
            corresponds to the generating sets in the theorem.
            \begin{itemize}[label=\textbullet, font=\LARGE]
                \item We show $\textit{Nexts}_{f_1 \cup f_2} =
                    \textit{Nexts}_{f_1} \cup \textit{Nexts}_{f_2}$,
                    and thus that $PHO_{f_1 \cup f_2}(PDel_1  \cup PDel_2) =
                    \textit{HOProd}(\textit{Nexts}_{f_1 \cup f_2}) =
                    \textit{HOProd}(\textit{Nexts}_{f_1} \cup \textit{Nexts}_{f_2})$
                    \begin{itemize}
                        \item Let $n \in Nexts_{f_1 \cup f_2}$. Then $\exists
                            q \in f_1 \cup f_2: obliv(q) = n$. By definition of
                            union, $q \in f_1$ or $q \in f_2$. We fix $q \in f_1$
                            (the case $q \in f_2$ is symmetric).
                            Then $n \in Nexts_{f_1}$.

                            We conclude that
                            $n \in Nexts_{f_1} \cup Nexts_{f_2}$.
                        \item Let $n \in Nexts_{f_1} \cup Nexts_{f_2}$.
                            We fix $n \in Nexts_{f_1}$
                            (as always, the other case is symmetric).
                            Then $\exists q \in f_1: obliv(q) = n$. As
                            $q \in f_1$ implies $q \in f_1 \cup f_2$,
                            we conclude that $n \in Nexts_{f_1 \cup f_2}$.
                    \end{itemize}
                \item We show $Nexts_{f_1 \combi f_2} = \{ n_1 \cap n_2 \mid
                    n_1 \in Nexts_{f_1} \land n_2 \in Nexts_{f_2}\}$.
                    \begin{itemize}
                        \item Let $n \in Nexts_{f_1 \combi f_2}$. Then
                            $\exists q \in f_1 \combi f_2: obliv(q) = n$.
                            By definition of combination, $\exists q_1 \in f_1,
                            \exists q_2 \in f_2: q_1.round = q_2.round= q.round
                            \land q = q_1 \combi q_2$. This means
                            $n = obliv(q) = obliv(q_1) \cap obliv(q_2)$.

                            We conclude that $n \in \{ n_1 \cap n_2 \mid
                            n_1 \in Nexts_{f_1} \land n_2 \in Nexts_{f_2}\}$.
                        \item Let $n \in \{ n_1 \cap n_2 \mid
                            n_1 \in Nexts_{f_1} \land n_2 \in Nexts_{f_2}\}$.
                            Then $\exists n_1 \in Nexts_{f_1},
                            \exists n_2 \in Nexts_{f_2}: n = n_1 \cap n_2$.
                            Because $f_1$ and $f_2$ are oblivious strategies, we
                            can find $q_1 \in f_1$ such that $obliv(q_1) = n_1$,
                            $q_2 \in f_2$ such that $obliv(q_2) = n_2$, and
                            $q_1.round = q_2.round$.

                            Then $q = q_1 \combi q_2$ is a state
                            of $f_1 \combi f_2$. We have $obliv(q) =
                            n_1 \cap n_2 = n$.

                            We conclude that $n \in Nexts_{f_1 \combi f_2}$.
                    \end{itemize}
                \item Trivially, $Nexts_f = Nexts_f$.
            \end{itemize}
        \end{proof}

  \subsection{Domination by an Oblivious Strategy}

      To prove Theorem~\ref{invarSuffCfree}, we first show that the condition
      implies the domination by an oblivious strategy.

      \begin{lemma}[Sufficient Condition to be Dominated by an Oblivious Strategy]
          \label{suffDomCarefr}
          Let $PDel$ be a delivered predicate. If
          \begin{itemize}
              \item \textbf{(Total collection)}
                  \textit{PDel} contains the total collection $c_{tot}$,
              \item \textbf{(Symmetry up to a round)}
                  $\forall c \in PDel, \forall p \in \Pi,
                  \forall r > 0, \forall r' > 0, \exists c' \in PDel:
                  c'[1,r'-1] = c_{tot}[1,r'-1] \land \forall q \in \Pi:
                  c'(r',q)=c(r,p)$
          \end{itemize}
          then $PDel$ is dominated by an oblivious strategy.
      \end{lemma}

      \begin{proof}
          Proved in \cite[Thm 24]{ShimiOPODIS18}.
      \end{proof}

        \begin{theorem*}[(\ref{invarSuffCfree})
                         Domination by Oblivious for Operations]
            Let $PDel, PDel_1, PDel_2$ be delivered predicates that satisfy:
            \begin{itemize}
              \item \textbf{(Total collection)}
                They contains the total collection $c_{tot}$,
              \item \textbf{(Symmetry up to a round)}
                $\forall c$ a collection in the predicate, $\forall p \in \Pi,
                \forall r > 0, \forall r' > 0, \exists c'$ a collection
                in the predicate:
                $c'[1,r'-1] = c_{tot}[1,r'-1] \land \forall q \in \Pi:
                c'(r',q)=c(r,p)$
            \end{itemize}
            Then $PDel_1 \cup PDel_2$, $PDel_1 \combi PDel_2$, $PDel_1 \leadsto
            PDel_2$, $PDel^{\omega}$ satisfy the same two conditions and
            are dominated by oblivious strategies.
        \end{theorem*}

        \begin{proof}[Proof idea]
            Thanks to Lemma~\ref{suffDomCarefr}, we only have to show that
            the condition is maintained by the operations; the domination
            by an oblivious strategy follows directly.

            For containing $c_{tot}$: $c_{tot} \cup c_{tot} = c_{tot}$;
            $c_{tot} \combi
            c_{tot} = c_{tot}$; $c_{tot} \leadsto c_{tot} = c_{tot}$;
            and the succession of $c_{tot}$ with itself
            again and again gives $c_{tot}$.

            As for symmetry up to a round, we show its invariance.
            Let $p \in \Pi, r > 0$ and $r' > 0$.
            \begin{itemize}
                \item If $c \in PDel_1 \cup PDel_2$, then $c \in PDel_1 \lor
                    c \in PDel_2$. We can then apply the condition for one of
                    them to get $c'$.
                \item If $c \in PDel_1 \combi PDel_2$, then
                    $\exists c_1 \in PDel_1, \exists c_2 \in PDel_2:
                    c = c_1 \combi c_2$. Applying the condition for
                    $c_1$ and $c_2$ gives us $c'_1$ and $c'_2$, and
                    $c' = c'_1 \combi c'_2$ satisfies the condition for $c$.
                \item If $c \in PDel_1 \leadsto PDel_2$, then
                    $\exists c_1 \in PDel_1, \exists c_2 \in PDel_2,
                    \exists r_{change} \geq 0: c = c_1[1,r_{change}].c_2$.
                    Applying the condition for $c_1$ at $r$ and $r'$
                    and for $c_2$ at $r-r_{change}$ and $r' - r_{change}$
                    gives us $c'_1$ and $c'_2$, and
                    $c' = c'_1[1,r_{change}].c'_2$ satisfies the condition for $c$.
                \item If $c \in PDel^{\omega}$, then
                    $\exists (c_i)_{i \in \mathbb{N}^*},
                    \exists (r_i)_{i \in \mathbb{N}^*},$
                    the collections and indices defining $c$.
                    Then let $i$ the integer such that $r \in [r_{i}+1,r_{i+1}]$.
                    Applying the condition for $c_{i'}$ at $r-r_{i'}$ and $r'-
                    r_{i'}$  with $i' \leq i$ gives us $c'_{i'}$, and
                    $c' = c'_1[1,r_2-r_1]\cdots{}c'_i[1,r_{i+1}-r_i]
                    .c_{i+1}[1,r_{i+2}-r_{i+1}]\cdots$ satisfies the condition for
                    $c$.
            \end{itemize}
        \end{proof}

        \begin{proof}
            Thanks to Lemma~\ref{suffDomCarefr}, we only have to show that
            the condition is maintained by the operations; the domination
            by an oblivious strategy follows directly.

            We first prove that $c_{tot}$ is still in the results of
            the operations.
            \begin{itemize}
                \item If $PDel_1$ and $PDel_2$ contain $c_{tot}$,
                    then $PDel_1 \cup PDel_2$ trivially contains it too.
                \item If $PDel_1$ and $PDel_2$ contain $c_{tot}$,
                    then $PDel_1 \combi PDel_2$ contains
                    $c_{tot} \combi c_{tot} = c_{tot}$.
                \item If $PDel_1$ and $PDel_2$ contain $c_{tot}$,
                    then $PDel_1 \leadsto PDel_2 \supseteq PDel_2$
                    contains it too.
                \item If $PDel$ contains $c_{tot}$,
                    we can recreate $c_{tot}$ by taking all $c_i = c_{tot}$
                    and whichever $r_i$. Thus, $PDel^{\omega}$ contains
                    $c_{tot}$.
            \end{itemize}

            Then we show the invariance of the symmetry up to a round.
            \begin{itemize}
                \item Let $c \in PDel_1 \cup PDel_2$. Thus $c \in PDel_1$
                    or $c \in PDel_2$. We fix $c \in PDel_1$ (the other
                    case is symmetric). Then for $p \in \Pi,r > 0$ and $r' > 0$,
                    we get a $c' \in PDel_1$.
                    satisfying the condition.
                    And since $PDel_1 \subseteq PDel_1 \cup PDel_2$,
                    we get $c' \in PDel_1 \cup PDel_2$.

                    We conclude that the condition still holds
                    for $PDel_1 \cup PDel_2$.
                \item Let $c \in PDel_1 \combi PDel_2$. Then
                    $\exists c_1 \in PDel_1, \exists c_2 \in PDel_2:
                    c = c_1 \combi c_2$. For $p \in \Pi,r > 0$ and $r' > 0$,
                    our hypothesis on $PDel_1$ and $PDel_2$ ensures
                    that there are $c'_1 \in PDel_1$ satisfying the
                    condition for $c_1$ and $c_2' \in PDel_2$ satisfying
                    the condition for $c_2$.

                    We argue that $c' = c_1' \combi c'_2$ satisfies
                    the condition for $c$. Indeed, $\forall r'' < r',
                    \forall q \in \Pi: c(r'',q) = c_1'(r'',q) \combi c_2'(r'',q)
                    = \Pi$ and $\forall q \in \Pi: c(r',q) = c_1'(r',q) \combi
                    c_2'(r',q) = c_1(r,p) \combi c_2(r,p) = c(r,p)$.

                    We conclude that the condition still holds for
                    $PDel_1 \combi PDel_2$.
                \item Let $c \in PDel_1 \leadsto PDel_2$. Since if $c \in
                    PDel_2$, the condition trivially holds by hypothesis,
                    we study the case where succession actually happens.
                    Hence, $\exists c_1 \in PDel_1, \exists c_2 \in PDel_2, \exists
                    r_{change} > 0: c = c_1[1,r_{change}].c_2$.
                    For $p \in \Pi, r > 0$ and $r'>0$, we separate two cases.
                    \begin{itemize}
                        \item if $r \leq r_{change}$, then
                            our hypothesis on $PDel_1$ ensures
                            that there is $c'_1 \in PDel_1$ satisfying the
                            condition for $c_1$.
                            We argue that
                            $c' = c_1'[1,r'].c_2 \in PDel_1 \leadsto PDel_2$
                            satisfies the condition for $c$.

                            Indeed, $\forall r'' < r', \forall q \in \Pi:
                            c'(r'',q) = c_1'(r'',q) = \Pi$, and $\forall q \in \Pi:
                            c'(r',q) = c_1(r,p) = c(r,p)$
                        \item if $r > r_{change}$, then
                            our hypothesis on $PDel_2$ ensures
                            that there is $c'_2 \in PDel_2$ satisfying
                            the condition for $c_2$ at $p$ and $r - r_{change}$.
                            That is, $c'_2[1,r'-1] = c_{tot}[1,r'-1] \land
                            \forall q \in \Pi: c'_2(r',q)=c_2(r-r_{change},p)$
                            We argue that $c' = c'_2 \in PDel_1 \leadsto PDel_2$
                            satisfies the condition for $c$.

                            Indeed, $\forall r'' < r', \forall q \in \Pi:
                            c'_2(r'',q) = \Pi$, and $\forall q \in \Pi:
                            c'_2(r',q) = c_2(r-r_change,p) = c(r,p)$
                    \end{itemize}

                    We conclude that the condition still holds for
                    $PDel_1 \leadsto PDel_2$.

                \item Let $c \in PDel^{\omega}$. Let $(c_i)$ and
                    $(r_i)$ be the collections and indices defining
                    $c$. We take $p \in \Pi,r > 0$ and $r' > 0$.
                    Let $i > 0$ be the integer such that
                    $r \in [r_i+1, r_{i+1}]$. By hypothesis on $PDel$,
                    There is $c'_i \in PDel$ satisfying the condition for $c_i$
                    at $p$ and $r - r_i$.
                    That is, $c'_i[1,r'-1] = c_{tot}[1,r'-1] \land
                    \forall q \in \Pi: c'_i(r',q)=c_i(r-r_i,p)$.

                    We argue that $c'_i \in PDel$
                    satisfies the condition for $c$. Indeed,
                    $\forall r'' \leq r', \forall q \in \Pi$, we have:
                    $c'_i(r'',q) = \Pi$ and $\forall q \in \Pi:
                    c'_i(r',q) = c_i(r-r_i,p) = c(r,p)$.

                    We conclude that the condition still holds for
                    $PDel^{\omega}$.
            \end{itemize}
        \end{proof}

\section{Proofs for Conservative Strategies}
\label{app:cons}

  \subsection{Minimal Conservative Strategies}

  We use a necessary and sufficient condition for an oblivious strategy to be valid in
  the rest of the proofs.

        \begin{lemma}[Necessary and Sufficient Condition for Validity of a Conservative Strategy]
            \label{consValid}
            Let $PDel$ be a delivered predicate and $f$ be a conservative strategy.
            Then $f$ is valid for $PDel \iff
            f \supseteq \{q \in Q \mid \exists c \in PDel, \exists p \in \Pi,
            \forall r \leq q.round: q(r) = c(r,p) \}$.
        \end{lemma}

        \begin{proof}
            From the version in~\cite{ShimiOPODIS18},
            $f$ has to satisfy
            $\forall \textit{CDel} \in \textit{PDel},
            \forall r > 0, \forall j \in \Pi:
            \langle r, \{ \langle r',k \rangle \mid r' \leq r
            \land k \in \textit{CDel}(r',j)\} \rangle \in \textit{Nexts}^R_f$.

            We show the equivalence of this condition with our own,
            which allow us to conclude by transitivity of equivalence.
            \begin{itemize}
                \item $(\Longrightarrow)$ assume our condition holds and
                    prove the one from~\cite{ShimiOPODIS18},

                    Let $c \in PDel, r > 0$ and $p \in \Pi$: we want to show that
                    $q = \langle r, \{ \langle r',k \rangle \mid r' \leq r
                    \land k \in c(r',p)\} \rangle
                    \in \textit{Nexts}^R_f$. That is to say, that all states
                    whose past and present correspond to this conservative state
                    are accepted by $f$.
                    Let $q'$ such that $cons(q') = q$, that is
                    $q'.round = q.round = r$ and $\forall r' \leq r:
                    q'(r') = q(r') = c(r',p)$. We have the collection
                    and the round to apply our condition, and thus $q \in f$.
                \item $(\Longleftarrow)$ assume the condition from~\cite{ShimiOPODIS18} holds and we prove ours.

                    Let $q$ such that $\exists c \in PDel,
                    \exists p \in \Pi, \forall r \leq q.round: q(r) = c(r,p)$.

                    Then $cons(q) =
                    \langle q.round, \{ \langle r,k \rangle \mid r \leq q.round
                    \land k \in c(r,p)\} \rangle$. This conservative state
                    is in $\textit{Nexts}^R_f$ by hypothesis.

                    We conclude that $q \in f$.
            \end{itemize}
        \end{proof}

        \begin{lemma*}[(\ref{minDom})
                       Domination of Minimal Conservative Strategy]
            Let $PDel$ be a delivered predicate and $f_{min}$
            be its minimal conservative strategy.
            Then $f_{min}$ dominates the conservative strategies for $PDel$.
        \end{lemma*}

        \begin{proof}
            First, $f_{min}$ is valid for $PDel$ by application of
            Lemma~\ref{consValid}.
            Next, we take another conservative strategy $f$, valid
            for $PDel$. Lemma~\ref{consValid} gives us that
            $f_{min} \subseteq f$. Hence, when
            $f_{min}$ allow a change of round, so does $f$.
            This entails that all executions of $f_{min}$ for
            $PDel$ are also executions of $f$ for $PDel$,
            and thus that the $PHO_{f_{min}}(PDel) \subseteq
            PHO_f(PDel)$.
        \end{proof}

    \subsection{Operations Maintain Minimal Conservative Strategy}

        \begin{theorem*}[(\ref{unionReac}) Minimal Conservative Strategy for Union]
            Let $PDel_1, PDel_2$ be two PDels, $f_1$ the minimal
            conservative strategy for $PDel_1$, and $f_2$ the
            minimal conservative strategy for $PDel_2$.
            Then $f_1 \cup f_2$ is the minimal conservative strategy for
            $PDel_1 \cup PDel_2$.
        \end{theorem*}

        \begin{proof}
            We only have to show that $f_1 \cup f_2$ is equal to
            Definition~\ref{minReac}.
            \begin{itemize}
                \item $(\supseteq)$ Let $q$ be a state such that
                    $\exists c \in PDel1 \cup PDel_2, \exists p \in \Pi$ such
                    that $\forall r \leq q.round: q(r) = c(r,p)$.
                    If $c \in PDel_1$, then $q \in f_1$,
                    because $f_1$ is the minimal conservative strategy
                    for $PDel_1$, and by application of Lemma~\ref{consValid}.
                    Thus, $q \in f_1 \cup f_2$.
                    If $c \in PDel_2$, the same reasoning apply with
                    $f_2$ in place of $f_1$.
                    We conclude that $q \in f_1 \cup f_2$.
                \item $(\subseteq)$ Let $q \in f_1 \cup f_2$.
                    This means that $q \in f_1 \lor q \in f_2$. The case where
                    it is in both can be reduced to any of the two.
                    If $q \in f_1$, then by minimality of $f_1$
                    $\exists c_1 \in PDel_1, \exists p_1 \in \Pi$ such that
                    $\forall r \leq q.round: q(r) = c_1(r,p_1)$.
                    $PDel_1 \subseteq
                    PDel_1 \cup PDel_2$, thus $c_1 \in PDel_1 \cup PDel_2$.
                    We found the $c$ and $p$ necessary to show $q$ is
                    in the minimal conservative strategy for $PDel_1 \cup PDel_2$.
                    If $q \in f_2$, the reasoning is similar to
                    the previous case, replacing $f_1$ by $f_2$
                    and $PDel_1$ by $PDel_2$.
            \end{itemize}
        \end{proof}

        \begin{theorem*}[(\ref{combiReac})
                        Minimal Conservative Strategy for Combination]
            Let $PDel_1, PDel_2$ be two symmetric PDels, $f_1$ the minimal
            conservative strategy for $PDel_1$, and $f_2$ the
            minimal conservative strategy for $PDel_2$.
            Then $f_1 \combi f_2$ is the minimal conservative strategy for
            $PDel_1 \combi PDel_2$.
        \end{theorem*}

        \begin{proof}
            We only have to show that $f_1 \combi f_2$ is equal to
            Definition~\ref{minReac}.
            \begin{itemize}
                \item $(\supseteq)$ Let $q$ be a state such that
                    $\exists c \in PDel1 \combi PDel_2, \exists p \in \Pi$ such
                    that $\forall r \leq q.round: q(r) = c(r,p)$.
                    By definition of $c$, $\exists c_1 \in PDel_1,
                    \exists c_2 \in PDel_2: c_1 \combi c_2 = c$.
                    We take $q_1$ such that $q_1.round = q.round$ and
                    $\forall r > 0:\\
                    \left(
                    \begin{array}{ll}
                        q_1(r) = c_1(r,p) &\text{if }r \leq q.round\\
                        q_1(r) = q(r) &\text{otherwise}\\
                    \end{array}
                    \right)$. We also take $q_2$ such that $q_2.round = q.round$
                    and $\forall r > 0:
                    \left(
                    \begin{array}{ll}
                        q_2(r) = c_2(r,p) &\text{if }r \leq q.round\\
                        q_2(r) = q(r) &\text{otherwise}\\
                    \end{array}
                    \right)$.

                    Then by validity of $f_1$ and $f_2$ (since they are minimal
                    conservative strategies) and by application of
                    Lemma~\ref{consValid}, we get $q_1 \in f_1$ and $q_2 \in f_2$.
                    We also see that $q = q_1 \combi q_2$. Indeed, for
                    $r \leq q.round$, we have $q(r) = c(r,p) = c_1(r,p) \cap c_2(r,p) =
                    q_1(r) \cap q_2(r)$; and for $r > q.round$, we have
                    $q(r) = q(r) \cap q(r) = q_1(r) \cap q_2 (r)$.

                    Therefore $q \in PDel_1 \combi PDel_2$.
                \item $(\subseteq)$
                    Let $q \in f_1 \combi f_2$. By definition of $f_1 \combi f_2$,
                    $\exists q_1 \in f_1, \exists q_2 \in f_2$ such that $q_1.round
                    = q_2.round = q.round$ and $q = q_1 \combi q_2$.

                    Since $f_1$ and $f_2$ are minimal conservative strategies of their
                    respective PDels, $\exists c_1 \in PDel_1, \exists p_1 \in \Pi$
                    such that $\forall r \leq q.round: q_1(r) = c_1(r,p_1)$; and
                    $\exists c_2 \in PDel_2, \exists p_2 \in \Pi$ such that
                    $\forall r \leq q.round: q_2(r) = c_2(r,p_2)$.

                    By symmetry of $PDel_2$, $\exists c'_2 \in PDel_2$
                    such that $\forall r \leq q.round: c'_2(r,p_1) = c_2(r,p_2)$.
                    Hence, $\forall r \leq q.round: q_2(r) = c'_2(r,p_1)$.

                    By taking $c = c_1 \combi c_2$, we get
                    $\forall r \leq q.round: q(r) = q_1(r) \cap q_2(r)
                    = c_1(r,p_1) \cap c_2(r,p_1) = c(r,p_1)$.

                    We found $c$ and $p$ showing that $q$ is in
                    the minimal conservative strategy for $PDel_1 \combi PDel_2$.
            \end{itemize}
        \end{proof}

        \begin{theorem*}[(\ref{succReac})
                         Minimal Conservative Strategy for Succession]
            Let $PDel_1, PDel_2$ be two symmetric PDels, $f_1$ the minimal
            conservative strategy for $PDel_1$, and $f_2$ the
            minimal conservative strategy for $PDel_2$.
            Then $f_1 \leadsto f_2$ is the minimal conservative strategy for
            $PDel_1 \leadsto PDel_2$.
        \end{theorem*}

        \begin{proof}
            We only have to show that $f_1 \leadsto f_2$ is equal to
            Definition~\ref{minReac}.
            \begin{itemize}
                \item $(\supseteq)$ Let $q$ be a state such that
                    $\exists c \in PDel1 \leadsto PDel_2, \exists p \in \Pi$ such
                    that $\forall r' \leq q.round: q(r') = c(r',p)$.
                    By definition of $c$, $\exists c_1 \in PDel_1,
                    \exists c_2 \in PDel_2, \exists r > 0: c = c_1[1,r].c_2$.
                    \begin{itemize}[label=\textbullet, font=\LARGE]
                        \item If $r = 0$, then $c[1,r] = c_2[1,r]$, and thus
                            $\forall r' \leq q.round: q(r') = c_2(r',p)$.
                            The validity of $f_2$ and Lemma~\ref{consValid}
                            then allow us to conclude that $q \in f_2$ and thus
                            that $q \in f_1 \leadsto f_2$.
                        \item If $r > 0$, we have two cases to consider.
                            \begin{itemize}
                                \item If $q.round \leq r$, then
                                    $\forall r' \leq q.round: q(r') = c_1(r',p)$
                                    We conclude by $f_1$ and application of
                                    Lemma~\ref{consValid} that $q \in f_1$ and thus
                                    that $q \in f_1 \leadsto f_2$.
                                \item If $q.round > r$, then $c[1,q.round]
                                    = c_1[1,r].c_2[1,q.round-r]$.

                                    We take $q_1$ such that $q_1.round = r$ and
                                    $\forall r' > 0:\\
                                    \left(
                                    \begin{array}{ll}
                                        q_1(r') = c_1(r',p) &\text{if }r' \leq q_1.round\\
                                        q_1(r') = q(r') &\text{otherwise }\\
                                    \end{array}
                                    \right)$. We also take $q_2$ such that
                                    $q_2.round = q.round - r$
                                    and $\forall r' > 0:
                                    \left(
                                    \begin{array}{ll}
                                        q_2(r') = c_2(r',p) &\text{if }r' \leq q_2.round\\
                                        q_2(r') = q(r'-q.round) &\text{otherwise}\\
                                    \end{array}
                                    \right)$.

                                    Then by validity of $f_1$ and $f_2$,
                                    and by application of Lemme~\ref{consValid},
                                    we get $q_1 \in f_1$ and $q_2 \in f_2$.
                                    We also see that $q = q_1 \leadsto q_2$.
                                    Indeed, for
                                    $r' \leq q_1.round = r$, we have $q(r') = c(r',p) =
                                    c_1(r',p) = q_1(r')$; for
                                    $r' \in [q_1.round + 1, q.round]$, we have
                                    $q(r') = c(r',p) = c_2(r'-r,p) = q_2(r'-r)$ and
                                    for $r' > q.round$ we have
                                    $q(r') = q_2 (r'-q.round)$.

                                    We conclude that $q \in f_1 \leadsto f_2$.
                            \end{itemize}
                    \end{itemize}
                \item $(\subseteq)$ Let $q \in f_1 \leadsto f_2$.
                    By definition of succession for strategies, there are
                    three possibilities for $q$.
                    \begin{itemize}
                        \item If $q \in f_1$, then by minimality of $f_1$
                            $\exists c_1 \in PDel_1, \exists p_1 \in \Pi:
                            \forall r \leq q.round: q(r) = c_1(r,p_1)$.
                            Let $c_2 \in PDel_2$. We take
                            $c = c_1[1,q.round].c_2$; we have $c \in c_1 \leadsto c_2$.

                            Then, $\forall r \leq q.round:
                            q(r)=c_1(r,p_1)=c(r,p_1)$.
                            We found $c$ and $p$
                            showing that $q$ is in the minimal conservative strategy for
                            $PDel_1 \leadsto PDel_2$.
                        \item If $q \in f_2$, then by minimality of $f_2$
                            $\exists c_2 \in PDel_2, \exists p_2 \in \Pi:
                            \forall r \leq q.round: q(r) = c_2(r,p_2)$.
                            As $PDel_2 \subseteq PDel_1 \leadsto PDel_2$, thus
                            $c_2 \in PDel_1 \leadsto PDel_2$.

                            We found $c$ and $p$
                            showing that $q$ is in the minimal conservative strategy for
                            $PDel_1 \leadsto PDel_2$.
                        \item There are $q_1 \in f_1$ and $q_2 \in f_2$ such that
                            $q = q_1 \leadsto q_2$.

                            Because $f_1$ and $f_2$ are the minimal conservative strategies
                            of their respective PDels, $\exists c_1 \in PDel_1,
                            \exists p_1 \in \Pi$ such that
                            $\forall r \leq q.round: q_1(r) = c_1(r,p_1)$; and
                            $\exists c_2 \in PDel_2, \exists p_2 \in \Pi$ such that
                            $\forall r \leq q.round: q_2(r) = c_2(r,p_2)$.

                            By symmetry of $PDel_2$, $\exists c'_2 \in PDel_2:
                            \forall r \leq q.round: c'_2(r,p_1) = c_2(r,p_2)$.
                            Hence, $\forall r \leq q.round: q_2(r) = c'_2(r,p_1)$.

                            By taking $c = c_1[1,q_1.round].c'_2$, we have
                            $c \in c_1 \leadsto c'_2$.
                            Then $\forall r \leq q.round
                            = q_1.round + q_2.round:\\
                            \left(
                            \begin{array}{ll}
                                \begin{array}{ll}
                                    q(r) & = q_1(r)\\
                                         & = c_1(r,p_1)\\
                                         & = c(r,p_1)\\
                                \end{array}
                                    & \text{if }r \leq q_1.round\\
                                \begin{array}{ll}
                                    q(r) & = q_2(r-q_1.round)\\
                                         & = c'_2(r-q_1.round,p_1)\\
                                         & = c(r,p_1)\\
                                \end{array}
                                    & \text{if }r \in [q_1.round+1,q_1.round+ q_2.round]\\
                            \end{array}
                            \right)$.

                            We found $c$ and $p$
                            showing that $q$ is in the minimal conservative strategy for
                            $PDel_1 \leadsto PDel_2$.
                    \end{itemize}
            \end{itemize}
        \end{proof}

        \begin{theorem*}[(\ref{repetReac})
                        Minimal Conservative Strategy for Repetition]
            Let $PDel$ be a symmetric PDel, and $f$ be its minimal conservative strategy.
            Then $f^{\omega}$ is the minimal conservative strategy for
            $PDel^{\omega}$.
        \end{theorem*}

        \begin{proof}
            We only have to show that $f^{\omega}$ is equal to
            Definition~\ref{minReac}.
            \begin{itemize}
                \item $(\supseteq)$ Let $q$ be a state such that
                    $\exists c \in PDel^{\omega}, \exists p \in \Pi$ such
                    that $\forall r \leq q.round: q(r) = c(r,p)$.
                    By definition of repetition, $\exists
                    (c_i)_{i \in \mathbb{N}^*}, \exists
                    (r_i)_{i \in \mathbb{N}^*}$ such that
                    $r_1 = 0$ and $\forall i \in \mathbb{N}^*:
                    (c_i \in PDel \land r_{i} < r_{i+1} \land
                    c[r_i+1,r_{i+1}]=c_i[1,r_{i+1} - r_i])$.

                    Let $k$ be the biggest integer such that $r_k \leq q.round$. We
                    consider two cases.
                    \begin{itemize}
                        \item If $r_k = q.round$, then
                            $c[1,r] = c_1[1,r_2-r_1].c_2[1,r_3-r_2]...c_{k-1}[1,r_k-r_{k-1}]$.
                            We take for $i \in [1,k-1]: q_i$ the state such that
                            $q_i.round = r_{i+1} -r_i$ and
                            $\forall r > 0:\\
                            \left(
                            \begin{array}{ll}
                                q_i(r) = c_i(r,p)
                                    &\text{if }r \leq q_i.round\\
                                q_i(r) = q(r+\sum\limits_{j \in [1,i-1]} q_i.round)
                                    &\text{otherwise }\\
                            \end{array}
                            \right)$.

                            By validity of $f$ and by application of
                            Lemma~\ref{consValid}, for $i \in [1,k-1]$
                            we have $q_i \in f$.
                            We see that $\forall r > 0: q(r) = (q_1 \leadsto ...
                            \leadsto q_{k-1})(r)$. Indeed, $\forall r \in
                            [r_i + 1, r_{i+1}]: q(r) = c(r,p) = c_i(r-r_i,p) =
                            q_i(r-r_i)$; and for $r > q.round: q(r) =
                            q_{k-1}(r-\sum\limits_{j \in [1,k-1]} q_i.round)$.

                            We conclude that $q \in f^{\omega}$.
                        \item If $q.round > r_k$, we can apply the same reasoning
                            as in the previous case, the only difference being
                            $c[1,r] =
                            c_1[1,r_2-r_1].c_2[1,r_3-r_2]...
                            c_{k-1}[1,r_k-r_{k-1}].c_k[1,r-r_k]$.
                    \end{itemize}
                \item $(\subseteq)$ Let $q \in f^{\omega}$.
                    By definition of $f^{\omega}$, $\exists q_1,q_2,...,q_k
                    \in f: q = q_1 \leadsto q_2 \leadsto ... \leadsto
                    q_k$.

                    Since $f$ is the minimal conservative strategy of $PDel$,
                    $\exists c_1, c_2, ..., c_k \in PDel, \exists
                    p_1,p_2,...,p_k \in \Pi:
                    \forall i \in [1,k]
                    q_i = \langle q_i.round, \{\langle r,j \rangle \mid
                    r \leq q_i.round \land j \in c_i(r,p_i) \}$.

                    By symmetry of $PDel$, for all $i \in [2,k]$,
                    $\exists c'_i \in PDel, \forall r \leq q_i.round:
                    c'_i(r,p_1) = c_i(r,p_i)$.

                    We take $c = c_1[1,q_1.round].c'_2[1,q_2.round]...
                    c'_{k-1}[1,q_{k-1}.round].c'_k$,
                    thus $c \in c_1 \leadsto c'_2 \leadsto ... \leadsto c'_k$.
                    Then $\forall r \leq q.round
                    = \sum\limits_{i \in [1,k]} q_i.round$, if
                    $r \in [\sum\limits_{i \in [1,i-1]} q_i.round + 1,
                            \sum\limits_{i \in [1,i]} q_i.round]$,
                    we have
                    $\left(
                    \begin{array}{ll}
                        q(r) & = q_i(r - \sum\limits_{i \in [1,i-1]} q_i.round)\\
                             & = c_i(r - \sum\limits_{i \in [1,i-1]} q_i.round,p_1)\\
                             & = c(r,p_1)\\
                    \end{array}
                    \right)$.

                    We found $c$ and $p$
                    showing that $q$ is in the minimal conservative strategy for
                    $PDel^{\omega}$.
            \end{itemize}
        \end{proof}

    \subsection{Computing Heard-Of Predicates}

        \begin{theorem*}[(\ref{upperBoundHO})
                           Upper Bounds on HO of Minimal Conservative Strategies]
            Let $PDel, PDel_1, PDel_2$ be PDels containing $c_{tot}$.
            Let $f^{cons}, f_1^{cons}, f_2^{cons}$ be their respective minimal
            conservative strategies, and $f^{obliv}, f_1^{obliv}, f_2^{obliv}$
            be their respective minimal oblivious strategies.
            Then:
            \begin{itemize}
                \item $PHO_{f_1^{cons} \cup f_2^{cons}}(PDel_1 \cup PDel_2)
                    \subseteq \textit{HOProd}(\textit{Nexts}_{f_1^{obliv}} \cup
                    \textit{Nexts}_{f_2^{obliv}})$.
                \item $PHO_{f_1^{cons} \leadsto f_2^{cons}}(PDel_1 \leadsto PDel_2)
                    \subseteq \textit{HOProd}(\textit{Nexts}_{f_1^{obliv}} \cup
                    \textit{Nexts}_{f_2^{obliv}})$.
                \item $PHO_{f_1^{cons} \combi f_2^{cons}}(PDel_1 \combi PDel_2)
                    \subseteq \textit{HOProd}( \{ n_1 \cap n_2 \mid
                    n_1 \in \textit{Nexts}_{f_1^{obliv}} \land
                    n_2 \in \textit{Nexts}_{f_2^{obliv}}\})$.
                \item $PHO_{(f^{cons})^\omega}(PDel^{\omega})
                    \subseteq \textit{HOProd}(\textit{Nexts}_{f^{obliv}})$.
            \end{itemize}
        \end{theorem*}

        \begin{proof}
            A oblivious strategy is a conservative strategy. Therefore,
            the minimal conservative strategy always dominates
            the minimal oblivious strategy. Hence, we get an upper bound
            on the heard-of predicate of the minimal conservative strategies
            by applying Theorem~\ref{oblivOpsHO}.
        \end{proof}

\end{document}